\documentclass[12pt]{article}
\usepackage[letterpaper, portrait, margin=1in]{geometry}
\usepackage[utf8]{inputenc}
\usepackage[numbers]{natbib}
\usepackage{nicefrac}
\usepackage{multirow}
\usepackage{tablefootnote}
\usepackage[dvipsnames,table]{xcolor}

\usepackage{float}

\usepackage{amsmath, mathtools, amsthm, amssymb}
\usepackage{algorithm}
\usepackage[noend]{algpseudocode}
\usepackage{times, yhmath}
\usepackage{multicol}
\usepackage{thmtools}
\usepackage{booktabs}

\usepackage{hyperref}
\hypersetup{
    colorlinks=true,
    linkcolor=blue,
    filecolor=blue,      
    urlcolor=blue,
    citecolor=blue,
    pdftitle={Notes - Topic},
    pdfpagemode=FullScreen,
    }

\usepackage[textwidth=4em,textsize=tiny]{todonotes}

\usepackage{cleveref}

\theoremstyle{plain}
\newtheorem{theorem}{Theorem}
\newtheorem{lemma}[theorem]{Lemma}
\newtheorem{claim}[theorem]{Claim}
\newtheorem{fact}[theorem]{Fact}
\newtheorem{corollary}[theorem]{Corollary}

\theoremstyle{definition}
\newtheorem{definition}[theorem]{Definition}

\theoremstyle{remark}
\newtheorem{remark}[theorem]{Remark}

\usepackage[most]{tcolorbox}
\newtcolorbox{idea}[1][]
{
colbacktitle=cyan,
colback=cyan!10,
arc=1pt,
boxrule=1pt,
title=#1 
}

\newtcolorbox{update}[1][]
{
colbacktitle=gray,
colback=gray!10,
arc=1pt,
boxrule=1pt,
title=#1 
}

\newtcolorbox{question}[1][]
{
coltitle=black,
colbacktitle=yellow,
colback=yellow!10,
arc=1pt,
boxrule=1pt,
title=#1 
}

\newtcolorbox{note}[1][]
{
coltitle=black,
colbacktitle=green,
colback=green!10,
arc=1pt,
boxrule=1pt,
title=#1 
}

\newtcolorbox{problem}[1][]
{
coltitle=black,
colbacktitle=red!60,
colback=red!10,
arc=1pt,
boxrule=1pt,
title=#1 
}

\newcommand{\sA}{\mathcal{A}}

\newcommand{\sO}{\mathcal{O}}
\newcommand{\sP}{\mathcal{P}}

\newcommand{\stream}{\mathcal{S}}

\newcommand{\streamlen}{m}
\newcommand{\streamuniverse}{\Omega}

\newcommand{\g}{\mathsf{g}}
\newcommand{\h}{\mathsf{h}}
\newcommand{\f}{\mathsf{f}}
\newcommand{\C}[1]{\mathcal{#1}}
\newcommand{\W}[1]{\mathbf{#1}}
\newcommand{\cc}[1]{#1}

\newcommand{\randtran}{\delta_{p}}
\newcommand{\Randtran}{\Delta_{p}}

\newcommand{\FM}[1]{\mathbb{F}_{#1}}
\newcommand{\field}[1]{\mathbb{#1}}

\newcommand{\Nat}{\field{N}}

\newcommand{\norm}[1]{\left\|{#1}\right\|}
\newcommand{\func}[3]{{#1} : {#2} \rightarrow {#3}}
\DeclarePairedDelimiter{\ceil}{\lceil}{\rceil}

\newcommand{\blue}[1]{\textcolor{blue}{#1}}

\newcommand{\sequence}[2]{{#1}_1,{#1}_2,...,{#1}_{#2}}
\newcommand{\fbrac}[1]{\left({#1}\right)}
\newcommand{\sbrac}[1]{\left\{{#1}\right\}}

\newcommand{\size}[1]{\left|{#1}\right|}

\newcommand{\poly}[1]{\mathrm{poly}\fbrac{#1}}

\newcommand{\bigo}[1]{O\fbrac{{#1}}}
\newcommand{\bigot}[1]{\widetilde{O}\fbrac{{#1}}}

\usepackage{lineno}
\usepackage{float}

\title{Computing over Data Streams using Catalytic Space}

\author{
Ripley Becker\\
University of Nebraska--Lincoln
\and
Sourav Chakraborty\\
Indian Statistical Institute, Kolkata
\and
Debarshi Chanda\\
Indian Statistical Institute, Kolkata
\and
A. Pavan\\
Iowa State University
\and
N.~V.~Vinodchandran\\
University of Nebraska--Lincoln
}
\date{}

\begin{document}

\maketitle

\begin{abstract}


We introduce a streaming model with \emph{catalytic memory}, an auxiliary workspace that must be returned to its initial state at the end of the computation. We show that catalytic space yields dramatic space savings for data stream algorithms.

We first study the exact computation of frequency moments in insertion-only data streams. For every $k\ge1$, we give an exact four-pass algorithm for computing $\mathbb{F}_{k}$ using $O(k\log m)$ clean space, where $m$ is the stream length. We also present a $(k+1)$-pass algorithm with the same clean-space complexity that uses a factor of $k$ less catalytic space than the four-pass algorithm.

For small moments, we obtain stronger results. In particular, we show that $\mathbb{F}_{2}$ and $\mathbb{F}_{3}$ can be computed exactly in two and three passes, respectively, using only $O(\log m)$ clean space.

Additionally, we show that exact $\mathbb{F}_{0}$ computation reduces to computing $\mathbb{F}_{k}$ for a suitably chosen large value of $k$, resulting in an exact four-pass algorithm for $\mathbb{F}_{0}$ using only $O(\log m)$ clean space. We further show how our frequency-moment algorithms can be used to exactly count induced occurrences of any fixed graph $H$ in a graph stream, yielding a four-pass algorithm that uses $O_H(\log n)$ clean space, where $n$ is the number of vertices in the graph. As a special case, we obtain an exact three-pass algorithm for triangle counting using $O(\log n)$ clean space.

All of our algorithms are multi-pass. We complement these algorithmic results with a matching limitation showing that catalytic memory does not provide additional power in the single-pass setting. Specifically, we prove that every randomized or deterministic single-pass streaming algorithm using $s$ bits of clean memory and catalytic space can be simulated in the standard streaming model, without catalytic memory, using $O(s)$ space.

\end{abstract}

\newpage
\section{Introduction}


We introduce data streaming algorithms that have access to additional {\em catalytic} memory. In the catalytic computation model, an algorithm is given, besides its ordinary workspace, an auxiliary memory region called {\em catalytic space}. At the beginning of the computation, the catalytic space is initialized to arbitrary content, and this space can be used during the computation. However, it is required that the contents of the catalytic space must be restored to their initial contents at the end of the computation. Since its contents are available but cannot be consumed, this memory behaves like a catalyst. Since its introduction in~\cite{Buhrman/STOC/2014/CatalyticIntro}, catalytic computation has emerged as a powerful model, with recent works showing that such reversible access to auxiliary memory can surprisingly enable efficient computation~\cite{AgarwalaM/focs/2025/BipartiteMatchingCatalyticLogspace,CookLMP/STOC/2025/StructureOfCatalyticSpace,CookP/ITCS/2026/EfficientCatalyticGraphAlgo,AlekseevFMSV/mfcs/2025/CatalyticBeyondLogDepth,Buhrman/TCS/2018/Catalytic:NonDetANdHierarchy,CookM/CCC/2022/TimevsSpaceCatalyticBranchingProgram,CookM/stoc/2020/CatalyticTreeEvaluation}. For example it is known that every language in nondeterministic logspace (NL) can be accepted by a deterministic logspace machine that has access to catalytic space.  A recent survey by Mertz provides an introduction to the topic~\cite{Mertz/EATCSBulletin/2023/CatalyticOpenProbs}. 

The main conceptual contribution of the present work is that such usable but non-consumable memory helps with 
data streaming computations too. In the standard streaming model, almost all non-trivial problems require randomness and approximation for space-efficient computation. For exact computation, the optimal deterministic/randomized algorithms essentially have to store the entire stream.
We show that catalytic space can circumvent this barrier: problems that are intractable in the deterministic streaming setting indeed become feasible when catalytic memory is available. 

In this paper, we focus on the problem of computing frequency moments of an insertion-only data stream.
For a data stream $\stream = \langle x_1,x_2,\cdots,x_m\rangle $ of items $x_i \in [n] = \{1,\cdots,n\}$, the $k^{th}$ frequency moment is the quantity $\FM{k}(\stream) = \sum_{i=1}^n \f_i^k$ where $\f_i$ is the frequency of item $i \in [n]$. Estimating frequency moments has been one of the central problems in the classical data-streaming model, beginning with the seminal work of Alon, Matias, and Szegedy~\cite{AMS}. The space complexity of this problem in the traditional streaming model is well understood. In particular, in the absence of randomness or approximation, the naive algorithm is essentially optimal from the perspective of space complexity: any algorithm for exactly computing $\FM{k}$ for any $k\neq 1$ requires $\Omega(n)$ space, even if {multiple passes over the stream or} randomization are allowed~\citep{AMS,RAZBOROV1992385}.

To illustrate the power of catalytic space in streaming computation, consider the problem of computing $\FM{2} = \sum_{i=1}^n \f_i^2$. We describe a simple three-pass algorithm that achieves this with $O(\log m)$ memory. We will think of catalytic space as $n$, $l$-bit registers $\C{C}[1],\C{C}[2],\cdots,\C{C}[n]$ with the initial content of $\C{C}[i]$ being $c[i]$. First, we compute sum $S_0 = \sum_{i} c[i]$ before starting  any computation over the data stream. Then in the first pass over the stream, as we see item $i \in [n]$, increase $\C{C}[i]$ by one. Thus, at the end of the pass, the  $i^{th}$ register will have $c[i] + \f_i$.  At this point we can compute $S_1 = \sum_i (c[i] + \f_i)^2$. In the second pass, we can reset $\C{C}[i]$ to its original content by reversing the process: if we see item $i$, {\em decrease} $\C{C}[i]$ by one. In the third pass we can compute $S_2 = \sum c[i]\f_i$ in the work space: when we see an item $i$, add $c[i]$ to the sum $S_2$. Note that in this pass, we are not modifying the content of the catalytic space.  Finally, $\FM{2} = S_1 - 2S_2 - S_0$. We can set $l = O(\log m)$ and do the computation modulo $2^l$ to implement the above 3-pass algorithm in $O(\log m)$ space. This idea can be extended with additional non-trivial work to compute $\FM{k}$ using $O(k\log m)$ space. However, the number of passes becomes exponential in $k$.  
The technical contribution of this work is to design catalytic algorithms for $\FM{k}$ with a much smaller number of passes.
At this point an obvious question is: can we get a non-trivial catalytic algorithm with just one pass? The answer is no: we show that a single pass catalytic streaming algorithm does not add any power above traditional deterministic or randomized algorithms. Thus, 2 or more passes are needed for non-trivial catalytic streaming computation.

Finally, we note that techniques developed for space-bounded Turing machine models do not generally carry over to streaming algorithms, where access to the input is far more restricted. Thus, although there has been substantial progress on space-bounded catalytic computation in the Turing machine setting, those techniques do not directly apply in the data-streaming model. New methods are therefore needed for our setting.

\subsection{Our Results}

Using a result from \cite{Buhrman/STOC/2014/CatalyticIntro} we design a 4 pass streaming algorithms for computing $\FM{k}$ for any $k\geq 2$. The algorithm is given in Section~\ref{sec:Fk_in_4_passes}. The 4-pass algorithm for computing $\FM{k}$ uses $O(nk)$ catalytic registers each register having $O(k\log mk)$ bits. We also present another simple algorithm  (Algorithm~\ref{alg:FkinKP+1}) for computing $\FM{k}$ for any $k\geq 2$ with $k+1$ passes that uses $n$ catalytic registers each with $O(k\log mk)$ bits, thus saving on the amount of catalytic space while using more passes.


Our first main set of results attempts to improve the number of passes required to compute $\FM{k}$: can $\FM{k}$ be computed in $< 4 $ passes in the catalytic streaming model? While we do not have a general algorithm for all $k$, we show that $\FM{k}$ can be computed in $k$ passes for $k = 2$ and $k=3$ (Algorithms \ref{alg:F2in2pass} and \ref{alg:F3in3pass}). This is achieved by exploiting the passes more effectively via a divide-and-conquer strategy.


Next, we consider the exact computation of $\FM{0}$ using catalytic space. 
We show that $\FM{0}$ can be computed using a simple reduction to computing $\FM{p}$ for some prime $p > m$ by employing Fermat's Little Theorem: that for a prime $p > m$, $\f_i^{p-1} = 1 \pmod p$ if $\f_i \neq 0$ and $0 \pmod p$ if $\f_i=0$. This gives a 4 pass algorithm for computing $\FM{0}$ with $O(\log m)$ bits of clean space. 
However, the number of catalytic registers required is $O(nm)$ with each register having $O(m\log m)$ bits.  
We show that for a certain restricted class of streams, we can compute $\FM{0}$ using a smaller number of catalytic registers while using a few more passes.  In particular, we consider {\em flat streams} in the sense that there is a small $t$ so that $\f_i \leq t$ for all $i$. We give a deterministic catalytic space algorithm to compute $\FM{0}$ for such streams with number of passes and the memory as a function of $t$. 

We also study subgraph counting over graph streams in the catalytic streaming model. By reducing subgraph counting to exact $\FM{k}$ computation, we show how to exactly count induced occurrences of any fixed graph $H$ in a graph stream. This yields a three-pass algorithm for exact triangle counting, based on our specialized algorithms for $\FM{2}$ and $\FM{3}$, and a four-pass algorithm for exact induced $H$-subgraph counting using our general $\FM{k}$ algorithm. For graph streams on $n$ vertices, both algorithms use $O(\log n)$ clean space, where the hidden constant in the latter result depends only on $H$.

Finally, we show that 1-pass catalytic algorithms do not have any advantage over the traditional streaming algorithms: any catalytic 1-pass deterministic (randomized) algorithm  that uses $s$ space can be simulated by a 1-pass traditional deterministic (respectively randomized) algorithm that uses $O(s)$ space. Thus two or more passes are needed for catalytic streaming algorithm to have any advantage over traditional streaming algorithms.

\begin{table}[ht!]
\centering
\renewcommand{\arraystretch}{1.2}
\begin{tabular}{llcccc}
\toprule
\textbf{Statistic} &  & \textbf{Non-Catalytic} & \textbf{Non-Catalytic} & \multicolumn{2}{c}{\textbf{Catalytic}}   \\
&  & \textbf{(Approximate)} & \textbf{(Exact)} & \multicolumn{2}{c}{\textbf{(Exact)}}  \\
\midrule
\multirow{3}{*}{$\FM{2}$}
& \emph{Passes} 
& 1 (randomized)
& $\ell$
& \multicolumn{2}{c}{2}  \\
& \emph{Clean Space} 
& $\Theta(\log n)$~\citep{AMS,BravermanZamir/STOC/2025/FrequcyMomentOptimality}
& $\Omega(n/\ell)$ \cite{RAZBOROV1992385}
& \multicolumn{2}{c}{$O(\log n)$}  \\
& \emph{Catalytic Space} 
& $N/A$
& $N/A$
& \multicolumn{2}{c}{$\bigot{n}$}  \\
& & &
& \multicolumn{2}{c}{\Cref{Thm: F2in2}}  \\
\midrule
\multirow{3}{*}{$\FM{3}$}
& \emph{Passes} 
& 1 (randomized)
& $\ell$
& \multicolumn{2}{c}{3}   \\
& \emph{Clean Space} 
& $\Theta(n^{1/3})$~\citep{IndykWoodruff/STOC/2005/FreqMomentOptimal,WoodruffZhou/ICALP/2021/MutliPassAndRandOrderFrequencyMoments}
& $\Omega(n/\ell)$ \cite{RAZBOROV1992385}
& \multicolumn{2}{c}{$O(\log n)$}  \\
& \emph{Catalytic Space} 
& $N/A$
& $N/A$
& \multicolumn{2}{c}{$\bigot{n}$}  \\
& & &
& \multicolumn{2}{c}{\Cref{Thm: F3in3}}  \\
\midrule
\multirow{4}{*}{$\FM{k}$}
& \emph{Passes} 
& 1 (randomized)
& $\ell$
& $4$ & $k+1$\\\
& \emph{Clean Space} 
& $\Theta\!\left(n^{1-\frac{2}{k}}\right)$~\citep{IndykWoodruff/STOC/2005/FreqMomentOptimal,WoodruffZhou/ICALP/2021/MutliPassAndRandOrderFrequencyMoments}
& $\Omega(n/\ell)$ \cite{RAZBOROV1992385}
& $O(k\log n)$ & $O(k\log n)$ \\
& \emph{Catalytic Space} 
& $N/A$
& $N/A$
& $\bigot{k^2n}$ & $\bigot{kn}$ \\
& 
& 
& 
& $\Cref{Thm: Fk in 4}$ & \Cref{thm:kplus1} \\

\midrule
\multirow{3}{*}{$\FM{0}$}
& \emph{Passes} 
& 1 (randomized)
& $\ell$
& \multicolumn{2}{c}{4}  \\
& \emph{Clean Space} 
& $\Theta(\log n)$~\citep{KaneNelsonWoodruff/PODS/2010/F0Optimal}
& $\Omega(n/\ell)$ \cite{RAZBOROV1992385}
& \multicolumn{2}{c}{$\bigo{\log n}$}  \\
& \emph{Catalytic Space} 
& $N/A$
& $N/A$
& \multicolumn{2}{c}{$\bigo{\poly{n}}$} \\
& 
& 
& 
& \multicolumn{2}{c}{\Cref{F0-in-4pass}} \\
\midrule
\multirow{3}{*}{ $\#H$}
& \emph{Passes}
& $O(1)$
& -- 
& \multicolumn{2}{c}{4} \\
& \emph{Clean Space}
& $O\left(\frac{m^{\rho(H)}}{\#H}\right)$\cite{AssadiKapralovKhanna/ITCS/2019/SimpleSUblinearSubgraph}
& -- 
& \multicolumn{2}{c}{$O_H(\log n)$\tablefootnote{$O_H(\cdot)$ denotes $c_HO(\cdot)$ where $c$ is a constant depending only on the target subgraph $H$.}} \\
& \emph{Catalytic Space}
& $N/A$
& $N/A$
& \multicolumn{2}{c}{$\tilde{O}_H(n)$} \\
& 
& 
& 
& \multicolumn{2}{c}{\Cref{Thm: Counting SUbgraphs}} \\
\bottomrule
\end{tabular}
\caption{Comparison of the space complexity of catalytic and non-catalytic streaming algorithms for frequency moments and subgraph problems. The last columns list upper bounds reported in the present work. For notational brevity, we assume $m = \mathrm{poly(n)}$ in the table for frequency moments. For subgraph counting, $n$, $m$ and $\#H$ are the number of vertices, edges, and the subgraph $H$ in the graph, and $\rho(H)$ is the fractional edge cover in the subgraph $H$.}
\label{tab:results}
\end{table}
\section{Preliminaries}

\subsection{Notation}

We assume that our streams are over the universe $[n] = \{1, 2, \cdots, n\}$ and use $m$ to denote the stream length. For a given stream, $\f_i$ denotes the number of times $i$ appears in the stream. The $k$th frequency moment of the stream is $\FM{k} = \sum_{i =1}^n \f_i^k$.
Note that $\FM{0}$ is the number of elements with non-zero frequencies---the number of distinct elements in the stream.

\begin{definition}[Catalytic Streaming Model]
In the catalytic streaming model, a streaming algorithm is given access to a read-only stream \(\stream=\langle x_1,\dots,x_m \rangle\) with $x_i \in [n]$, a workspace, and a catalytic memory. The workspace is initialized to an all-zero vector, while the catalytic memory contains an arbitrary string $z$. The algorithm may use and modify the catalytic memory during the computation, but it is required to restore it to its initial value $z$ at the end of the computation. A catalytic streaming algorithm computes a function over $f$ if for every input stream $\stream$ and every initial catalytic memory content $z$, the algorithm outputs the correct value of $f(\stream)$ and terminates with catalytic memory content exactly $z$. If the algorithm scans the stream sequentially for $p$ passes, uses $s$-bits of workspace, it is called a $p$-pass, $s$-space catalytic streaming algorithm.
\end{definition}

A catalytic algorithm has two types of memory spaces: We refer to the workspace as {\em clean space} and the total amount of catalytic memory as {\em catalytic space}. We measure the space in terms of bits. The goal is to design algorithms that use small clean as well as catalytic space.

Next we describe the notation used in this paper (most of it is given in Table~\ref{tab:notations}). In some of our algorithms, we view $\FM{k}$ as the $k^{th}$ power of $\ell_k$-norm of the frequency vector and hence we also use $\norm{\f}_k^k$ for the $k^{th}$ moment. When it is relevant, we also specify the stream $\stream$ and use $\FM{k}(\stream)$ to denote the $k^{th}$ moment of $\stream$. We use calligraphic uppercase letters to denote the catalytic space. We view the catalytic space divided into registers of appropriate length and index them with $i$. We use the lowercase alphabet to denote the initial contents of the catalytic space. For example, the initial contents of $\C{A}[i]$ are $a[i]$.  We assume that the stream length $m$ is known to the algorithm a piori and assume $m \geq n$ for stating our bounds.


\begin{table}[t]
\centering
\renewcommand{\arraystretch}{1.15}
\begin{tabular}{ll}
\toprule
\textbf{Concept} & \textbf{Notation} \\
\midrule
Catalytic space & $\C{A}[i], \C{B}[i], \C{C}[i], \C{D}[i]$ \\
Initial catalytic contents & $\cc{a}[i], \cc{b}[i], \cc{c}[i], \cc{d}[i]$ \\
Frequency Vector & $\f$ \\
$k$th frequency moment & $\norm{\f}_k^k$, $\FM{k}$ \\
Frequency of the $i$th element & $\f_i$ \\
Workspace memory & $\W{S}$ \\
Stream & $\mathcal{S} = \langle \sequence{x}{m} \rangle$ \\
Universe & $[n] = \{1,2,\dots,n\}$ \\
\bottomrule
\end{tabular}
\caption{Notation used in the paper.}
\label{tab:notations}
\end{table}

\vspace{2mm}

\noindent{\em {Modulo Arithmetic:}} 
As mentioned earlier, we use the catalytic memory as a collection of registers of fixed size, each capable of storing \(l\) bits. Our algorithms perform additions and subtractions on these registers, so intermediate values may overflow. To handle this, all arithmetic is carried out modulo \(M = 2^l\). We choose \(l\) large enough so that the final quantity to be computed fits within \(l\) bits in the worst case. Consequently, if the desired output is \(a < M\), then
$a \bmod M = a,$
and hence the modular computation recovers the correct value.

\paragraph{Organization.} The rest of the paper is organized as follows: In \Cref{sec:overview}, we provide intuition and high-level proof details of our results. \Cref{sec:upper} contains formal proofs of upper bound results for moment computation: Computing $\FM{K}$ is $k+1$ passes, $\FM{2}$ in $2$ passes, $\FM{3}$ is 3 passes and computing $\FM{0}$ for flat streams. \Cref{subsec:subgraphcounting} we give proofs for subgraph counting. In \Cref{sec:lower}, we establish the lower bound results that catalytic space does not help for 1-pass streaming algorithms.
\section{Main Results and Proof Ideas} \label{sec:overview}

\subsection{\texorpdfstring{$\FM{k}$ in $4$-passes}{F4 in 4 passes}}\label{sec:Fk_in_4_passes}

In this section, we present a four-pass streaming algorithm to exactly compute the $k$-th frequency moment.

\begin{theorem}\label{Thm: Fk in 4}
    There exists a $4$-pass streaming algorithm using $\bigo{k\log m}$ clean space, and $\bigo{k^2n\log m}$ catalytic space that computes $\FM{k}$ exactly for any integer $k \geq 1$.
\end{theorem}
Our starting point is the powering lemma of \citet{Buhrman/STOC/2014/CatalyticIntro}. We call a program $\sP$ using catalytic registers to be reversible if there exists $\sP^{-1}$ that reverses the operation it performs and restores all the catalytic registers.

\begin{lemma}[Powering Lemma~\citep{Buhrman/STOC/2014/CatalyticIntro}]\label{Lem: Powering Lemma}
    Let $k$ be a positive integer, and $\sP$ be a reversible program with access to catalytic register $\C{C}$, and read-only access to some input that computes:
    \begin{align*}
        \C{C} = \cc{c}+x
    \end{align*}
    Then, there exists reversible programs $I_1,I_2$ and $I_3$ with access to $O(k)$ catalytic registers, in addition to catalytic register $\C{C}$, and $\C{D}$ such that running $I_1,P,I_2,P^{-1},I_3$ computes:
    \begin{align*}
        \C{D} = \cc{d}+x^k
    \end{align*}
\end{lemma}

We compute the $\FM{k}$ using four passes by storing each of the element-wise frequencies, and use \Cref{Lem: Powering Lemma} on these frequencies.


\begin{proof}[Proof of \Cref{Thm: Fk in 4}]
    We first note that given a pass over the input stream, we can define $\sP_i$ to be the program that adds $1$ to its catalytic register $\C{C}[i]$ at every occurrence of the element $x_i$ in the stream. Hence, at the end of the pass, it computes:
    \begin{align*}
        \C{C}[i] = \cc{c}_i + \f_i
    \end{align*}
    The program $\sP_i$ is clearly reversible, and $\sP_i^{-1}$ is the program that subtracts $1$ from its catalytic register $\C{C}[i]$ at every occurrence of the element $x_i$ in the stream. Now, for each $i \in [n]$, we use $\sP_i$ and \Cref{Lem: Powering Lemma} to compute:
    \begin{align*}
        \C{D}_i = \cc{d}_i + \f_i^k
    \end{align*}
    We maintain the sum of all $D_i$s to as $\widehat{\mathrm{F_k}}$:
    \begin{align*}
        \widehat{\mathrm{F_k}} = \sum_{i \in [n]} \fbrac{\cc{d}_i + \f_i^k}
    \end{align*}
    Using the fact that the programs used in \Cref{Lem: Powering Lemma} are reversible, we restore all catalytic registers to their initial values by executing $I_3^{-1},\sP_i,I_2^{-1},\sP_i^{-1},I_1^{-1}$ for all $i \in [n]$. Now, observe that, we have restored the catalytic registers $\C{D}[i]$ to their original values, i.e.
    \begin{align*}
        \C{D}[i] = \cc{d}_i
    \end{align*}
    Now, we subtract all $D_i$s from $\widehat{\mathrm{F_k}}$ to obtain the $k$-th frequency moment of the stream, i.e.
    \begin{align*}
        \widehat{\mathrm{F_k}} = \sum_{i \in [n]} \fbrac{\f_i^k}
    \end{align*}
    Observe that executing \Cref{Lem: Powering Lemma} for each $i \in [n]$ requires $\bigo{k}$ registers, and each register can contain at most $m^k$. Hence, the catalytic space used by the algorithm is $\bigo{k^2n\log m}$. The usage of clean space comes from storing the frequency moment itself, which is at most $m^k$. Furthermore, each execution of $\sP_i$ and $\sP_i^{-1}$ uses one pass over the stream, and can be done in parallel for all $i \in [n]$. Hence, the algorithm uses four passes over the stream as $\sP_i$ and $\sP_i^{-1}$ are each executed twice for all $i \in [n]$.
\end{proof}

\subsection{\texorpdfstring{Computing $\FM{k}(\stream)$ in $k+1$ pass with Improved Catalytic Space}{Computing F\_k(S) in k+1 pass}}

In this section, we present a $(k+1)$-pass, $O(k \log mk)$-space catalytic streaming algorithm to compute $\FM{k}(\stream)$. First, we describe an approach that is based on the ideas from the introduction. We will illustrate with $\FM{3}$. As in the case of algorithm for $\FM{2}$, we can modify $\C{C}[i]$ to $c[i]+\f_i$ in one pass, and compute $\sum_{i} (c[i]+\f_i)^3$. Now we need to subtract the following terms (i) $\sum_i c[i]^3$, (ii) $3\sum_{i} c[i]^2\f_i$, and (iii) $3 \sum_{i}c[i]\f_i^2$. The first term can be computed a priori and the second term can be computed in the beginning with a pass (without changing the catalytic memory). However, the third term involves a generalized second moment computation, i.e, compute $\FM{2}$ of the stream where the $i$-th element appears $\sqrt{c[i]}\f_i$ times, which require 3-passes. Finally, we need an additional pass to restore the contents of the catalytic space. This leads to an algorithm with 6 passes (though we have to handle subtle issues such as $c[i]$ not being a perfect square). In general, this approach, if could be extended for general $k$, would lead to an algorithm with passes exponential in $k$.  We do not take this direction. Instead, we employ a drastically different approach that involves maintaining a global sum (in the workspace) and updating it with alternating additive and subtractive terms which in the ends adds to $k!\FM{k}(\stream)$. The details are in \Cref{Subsec: Fk in k+1}.

\subsection{\texorpdfstring{Computing $\FM{2}$ and $\FM{3}$ in $2$ and $3$  Passes (respectively)}{Computing F\_2 and F\_3 in 2 and 3  Passes (respectively)}}

The main idea for computing the second frequency moment $\FM{2}(\stream)$ in two passes is based on a recursive decomposition of the data stream. Consider splitting the stream $\stream$ into two halves based on their arrival order. The first half is denoted $\stream_L$ and the second half is denoted $\stream_R$, with corresponding frequency vectors $\g$ and $\h$. Since the overall frequency vector satisfies $\f = \g + \h$, we can expand
\[
\FM{2}(\stream) = \sum_i \f_i^2 = \sum_i (\g_i + \h_i)^2 = \FM{2}(\stream_L) + \FM{2}(\stream_R) + 2 \sum_i \g_i \h_i.
\]
Thus, computing $\FM{2}(\stream)$ reduces to computing $\FM{2}$ on the two halves (which can be done recursively) and evaluating the cross term $\sum_i \g_i \h_i$. Hence, the main challenge is to compute the inner product $\sum_i \g_i \h_i$ in two passes. For this we use a catalytic space $\C{C}$ let $c[i]$ denote the initial contents of $\C{C}[i]$.

In the first pass over $\stream_L$, whenever we see $i$, we increment  $\C{C}[i]$ by 1. Thus after processing $\stream_L$, $\C{C}[i] = c[i] + \g_i$. Now while processing $\stream_R$, we compute $S = \sum_i(c[i]+\g_i)\h_i$ in the non-catalytic space. This is done by incrementing a counter by $\C{C}[i]=(c[i]+\g_i)$ when we see $i$.  This is the end of pass 1. Now, in the second pass, while processing $\stream_L$, we restore $\C{C}[i]$ to $c[i]$ by decrementing $\C{C}[i]$ by one when we see $i$.  In the same pass, while processing $\stream_R$, we compute $\sum_i c[i]\h_i$ (by incrementing a counter by $c[i]$ when we see $i$). This is stored in the non-catalytic space. At the end, we obtain $\sum_i \g_i\h_i$ by subtracting $\sum_i c[i]\h_i$ from $\sum_i(c[i]+\g_i)\h_i$. Combining this idea with recursive computations allows use to obtain $\FM{2}(\stream)$ in two passes.

The main idea for computing  $\FM{3}(\stream)$ is similar to that for computing $\FM{2}(\stream)$, in that both rely on a recursive decomposition of the data stream. In the case of $\FM{3}(\stream)$, we observe that
\[
\FM{3}(\stream) = \sum_i \f_i^3 = \sum_i (\g_i + \h_i)^3 
= \FM{3}(\stream_L) + \FM{3}(\stream_R) 
+ 3 \sum_i \g_i^2 \h_i + 3 \sum_i \g_i \h_i^2.
\]
Thus, computing $\FM{3}(\stream)$ reduces to evaluating the cross terms 
$\sum_i \left(\g_i^2 \h_i + \g_i \h_i^2\right)$.

Using ideas similar to those for computing $\FM{2}(\stream)$, we can load $\g_i$ and $\h_i$ into two separate catalytic spaces, say $\C{C}_i$ and $\C{D}_i$, so that they store $(c[i] + \g_i)$ and $(d[i] + \h_i)$ respectively. In a subsequent pass, we compute the quantities
\[
\sum_i (c[i] + \g_i)^2 \h_i 
\quad \text{and} \quad 
\sum_i (d[i] + \h_i)^2 \g_i.
\]
Expanding these expressions introduces additional terms, namely 
\[
\sum_i c[i]^2 \h_i, \quad \sum_i d[i]^2 \g_i, 
\quad \text{and} \quad \sum_i (c[i] + d[i]) \g_i \h_i,
\]
which must be subtracted out to isolate the desired cross terms.

In further passes, we can restore the catalytic spaces and compute $\sum_i c[i]^2 \h_i$ and $\sum_i d[i]^2 \g_i$. The remaining mixed term $\sum_i (c[i] + d[i]) \g_i \h_i$ can be handled using the techniques developed for computing $\FM{2}(\stream)$. This yields an algorithm that computes $\FM{3}(\stream)$ in $4$ passes. However, with a more careful and interleaved use of catalytic space, the number of passes can be reduced to $3$. We now describe this. In the below, we view a pass $j$ as two parts: $j_L$ processes $\stream_L$, and part $j_R$ that processes $\stream_R$.  We use four blocks of catalytic spaces $\C{A}$, $\C{B}$, $\C{C}$, and $\C{D}$. 

\smallskip
\noindent{\bf Pass $1_L$:} Update $\C{C}[i]$ to $c[i]+\g_i$ and $\C{B}[i]$ to $b[i] + a[i]\g_i$.

\noindent{\bf Pass $1_R$:} Update $\C{A}[i]$ to $a[i]+\h_i$. Compute $X_1 = \sum_i (c[i]+\g_i)^2\h_i$, $X_2 = \sum_i(b[i]+a[i]\g_i)\h_i$,

\noindent{\bf Pass $2_L$:} Restore $\C{C}[i]$ to $c[i]$. Compute $X_3 = \sum_i (a[i]+\h_i)^2\g_i$, and $X_4 = \sum_i d[i]\g_i$.

\noindent{\bf Pass $2_R$:} Update $\C{D}[i]$ to $d[i] + c[i]\h_i$. Restore $\C{A}[i]$ to $a[i]$. Compute $X_5 = \sum_i c[i]^2\h_i$.

\noindent{\bf Pass $3_L$:} Restore $\C{B}[i]$ to $b[i]$. Compute $X_6 = \sum_i (d[i]+c[i]\h_i)\g_i$, and $X_7 = \sum a[i]^2\g_i$. 

\noindent{\bf Pass $3_R$:} Restore $\C{D}[i]$ to $d[i]$ and compute $X_8 = \sum_i b[i]\h_i$.

\smallskip
Now, note that $X_1 -X_5 - 2 (X_6 -X_4)$ equals $\sum_i \g_i^2\h_i$. Similarly, $X_3 - X_7 - 2(X_2 - X_8)$ equals $\sum_i\h_i^2\g_i$. Combining this idea with a recursive approach enables us to compute $\FM{3}$ in 3 passes. The algorithms are presented in detail in Section~\ref{Section: 2-pass f_2}~and~\ref{Section: F3 in 3passes}.


\subsection{\texorpdfstring{$\FM{0}$ in $4$ passes}{F\_0 in 4 Passes}}

Building upon the 4-pass algorithm for $F_k$, we design an algorithm for $\FM{0}$ is 4 passes. We achieve this by reducing $\FM{0}$ to computing $\FM{k}$ for higher moments using Fermat's Little Theorem.

\begin{theorem}[Fermat's Little Theorem]
Let \(p\) be a prime number and let \(a\) be an integer. Then
\[
a^{p-1} \equiv
\begin{cases}
1 \pmod p, & \text{if } p \nmid a,\\[4pt]
0 \pmod p, & \text{if } p \mid a.
\end{cases}
\]
\end{theorem}

Our algorithm works as follows. Pick a prime $p$ that is larger than $\max\{m, n\}$. We compute $\FM{p-1}$ using 4 passes and output $\FM{p-1} \pmod p$. We claim that $\FM{p-1} \pmod p$ equals $\FM{0}$. 

\begin{claim}\label{Claim: Fp equals F0 mod p}
    $\FM{p-1} \pmod p = \FM{0} \pmod p  $
\end{claim}

\begin{proof}
Since $p > \max\{m, n\}$, we have $\f_i < p$. Thus from Fermat's little theorem, $\f_i^{p-1} = 1 \pmod p$ if $\f_{i} \neq 0$ and $0 \pmod p$ if $\f_i= 0$. 
Thus
\[\FM{0} = |\{i\mid \f_i^{p-1} \pmod p = 1\}| \] 
Equivalently 
\[\FM{0} = \sum_{i=1}^n \left[\f_{i}^{p-1} \pmod p\right]\]
where the sum is taken over the integer ring.
 Now,
\begin{eqnarray*}
    \FM{p-1} \pmod p  & = & \left(\sum_{i=1}^n \f_{i}^{p-1} \right) \pmod p\\
    & = & \sum_{i=1}^n \left (\f_{i}^{p-1} \pmod p\right) \pmod p\\
    & = & \FM{0} \pmod p
\end{eqnarray*}

This completes the proof of the claim.
\end{proof}

\begin{theorem}\label{F0-in-4pass}
There is a 4-pass streaming algorithm using $O(\log m)$ clean space and $O(m \cdot n)$ catalytic space.
\end{theorem}

\begin{proof}
    Since $p > \FM{0}$, $\FM{0} \pmod p = \FM{0}$. Thus $\FM{p-1} \pmod p$ equals $\FM{0}$. By Theorem~\ref{Thm: Fk in 4}, we can compute $\FM{p-1}$ in 4-passes, using $O(\log n)$ clean space and $O(pn)$ catalytic space. 

Note that we do \emph{not} compute $\FM{p-1}$ exactly, as doing so would require $O(p\log n)$ clean space. In the four-pass algorithm Theorem~\ref{Thm: Fk in 4}, the only clean memory is used to maintain the running sum of the values $f_i^k$. Thus, it suffices to store this accumulator modulo $p$, requiring only $O(\log p)$ clean space, while the catalytic computation proceeds unchanged modulo $p$. Consequently, the algorithm computes $\FM{p-1}\pmod p$, and by \Cref{Claim: Fp equals F0 mod p}, this uniquely determines $\FM{0}$.
\end{proof}

In the above theorem the number of catalytic registers required is $O(nm)$.  
Can the number of catalytic registers required be reduced?
We show that for a certain restricted class of streams, we can compute $\FM{0}$ using a smaller number of catalytic registers while using a few more passes.  In particular, we consider {\em flat streams} in the sense that there is a small $t$ so that $\f_i \leq t$ for all $i$. We give a deterministic catalytic space algorithm to compute $\FM{0}$ for such streams with number of passes and the memory as a function of $t$. 
For example, for $\mathrm{poly}(\log m)$-flat streams,  $\FM{0}$ can be computed in $\mathrm{poly}(\log m)$-space using ${\rm poly}(\log m)$ passes.  This is done via reducing computation of $\FM{0}$ to computation of a number of higher-frequency moments and then applying the Chinese remainder theorem. 
Note that computing $\FM{0}$ in the traditional streaming model, even for streams that are 2-flat ($\f_i \leq 2$) with $\mathrm{poly}(\log m)$ passes will require $\Omega(n/\mathrm{poly}(\log m))$ space~\cite{RAZBOROV1992385}. The details are provided in Section~\ref{Sec: Tflat}.

\subsection{Subgraph counting in Graph Streams}
Our algorithms for exact frequency moment computation also yield exact
algorithms for subgraph counting in graph streams.
The first consequence is an exact triangle counting algorithm, obtained
using the classical reduction of Bar-Yossef, Kumar, and Sivakumar from
triangle counting to frequency moments.
More generally, we show that every fixed induced subgraph can be counted
exactly in the catalytic streaming model. We present the results in detail in Section~\ref{subsec:subgraphcounting}.

\begin{theorem}[Informal]
For every fixed graph $H$, the number of induced copies of $H$ in a graph
stream can be computed exactly in four passes using $O_H(\log n)$ clean
space. As a special case, the number of triangles can be computed exactly
in three passes using $O(\log n)$ clean space.
\end{theorem}

The reduction associates a virtual stream with the input graph stream.
For triangle counting, every arriving edge $(u,v)$ generates the triples
$\{u,v,w\}$ for all $w\neq u,v$. Each triple appears once for every edge
it induces, so its frequency equals the number of edges among the three
vertices. Consequently, the third frequency moment isolates the triangles,
allowing the triangle count to be recovered from $\FM{1},\FM{2}$, and
$\FM{3}$.

Our generalization follows the same philosophy but uses embeddings of the
pattern graph instead of vertex subsets. Let $H$ be a fixed graph with
$q$ edges. The virtual stream contains one item for every injective
embedding of $H$ into the vertex set of the input graph. Whenever an edge
of the graph stream arrives, we update precisely those embeddings whose
corresponding edge in $H$ has been realized. Thus the frequency of an
embedding equals the number of edges of $H$ that have appeared in the
input graph. An embedding is a copy of $H$ precisely when its frequency
equals $q$.

To detect this event, we use the polynomial
\[
\binom{x}{q},
\]
which equals $1$ when $x=q$ and $0$ for every integer
$0\le x<q$. Since this polynomial has degree $q$, the number of copies of
$H$ is a fixed linear combination of the moments
$\FM{0},\FM{1},\ldots,\FM{q}$.
Finally, induced copies are obtained from the non-induced counts by a
constant-size inclusion--exclusion over the supergraphs of $H$. Since
$H$ is fixed, all required moments can be computed simultaneously using
our four-pass $\FM{k}$ algorithm, giving an $O_H(\log n)$-space
algorithm.

\subsection{Lower Bounds for 1-pass Catalytic Streaming Algorithms}

A few questions that arise from the upper-bound results: whether multiple passes are necessary for the exact computation of the frequency moments.  Can we reduce the number of passes required to 1, either by employing randomization or settling for approximations?

In this section, we show that catalytic space provides no advantage in the one-pass streaming model. In particular, any problem that can be solved by a deterministic (randomized) one-pass streaming algorithm using \(s\) bits of workspace and \(c\) bits of catalytic space can also be solved by a standard deterministic (respectively, randomized) one-pass streaming algorithm using only \(s\) bits of workspace. To prove this, we use an automata-theoretic view of streaming algorithms. The proof is inspired by the proof that in the traditional complexity setting Catalytic LOGSPACE is contained in the probabilistic class ZPP~\cite{Buhrman/STOC/2014/CatalyticIntro}.

 A deterministic one-pass streaming algorithm with \(s\) bits of workspace and \(c\) bits of catalytic space can be viewed as an automaton with \(2^{s+c}\) states, where each state represents a complete configuration of the workspace together with the catalytic memory. By contrast, a standard one-pass streaming algorithm with \(s\) bits of workspace corresponds to an automaton with \(2^s\) states.

In the catalytic setting, the initial contents of the catalytic memory may be arbitrary. Thus, there are \(2^c\) possible initial catalytic configurations, and hence \(2^c\) possible start states of the automaton. The catalytic requirement is that, regardless of the computation, the catalytic memory must be restored to its initial value at the end of the stream. We exploit this restoration property to show that the sets of states reachable from two different initial catalytic configurations are disjoint.

This immediately yields the desired bound. Since the full automaton has only \(2^{s+c}\) states, and the reachable sets corresponding to the \(2^c\) possible initial catalytic configurations are pairwise disjoint, it follows by averaging that for at least one initial catalytic configuration the number of reachable states is at most
$
\frac{2^{s+c}}{2^c}=2^s.
$
Fix such an initial catalytic configuration. Restricting the automaton to the states reachable from this configuration gives a standard finite automaton with at most \(2^s\) states. Therefore, it can be implemented as a standard one-pass streaming algorithm using only \(s\) bits of workspace and no catalytic space. Using this equivalence, we obtain the following result.

\begin{theorem}\label{thm:1passFk}
    Any 1-pass, deterministic, catalytic streaming algorithm that approximates $\FM{k}$ ($k \neq 1$) requires $\Omega(n)$-space. Any 1-pass, randomized, catalytic streaming algorithm that exactly computes $\FM{k}$ ($k \neq 1$) requires $\Omega(n)$-space. 
\end{theorem}

We present the results in detail in Section~\ref{sec:lower}.

\section{Open Problems}

We introduced a model of streaming computation with \emph{catalytic memory}: auxiliary memory that may be used throughout the computation, but must be restored to its original contents at the end. We showed that this form of usable but unconsumable memory can dramatically enhance the power of streaming algorithms. In particular, it enables space-efficient deterministic algorithms for the exact computation of frequency moments, yielding exponential improvements over the traditional streaming model, where brute-force algorithms are typically space-optimal. We like to note that our 4 pass and \(k+1\)-pass algorithms are stronger than a standard multi-pass streaming algorithms. In the usual multi-pass model, every pass sees the stream in the same order. In contrast, our algorithm for \(\FM{k}\) continues to work even if, on each pass, an adversary presents the elements in an arbitrary order, as long as the underlying stream itself is unchanged. Our results strongly indicate that catalytic space can be powerful resource for space-efficient exact deterministic computation over data streams, and opens new directions for research in the area. 

Although we showed that \(\FM{k}\) can be computed in 4 passes, it is unclear whether this is necessary. It is plausible that, for every constant \(k\), \(\FM{k}\) can in fact be computed in 2 passes. Our lower bound only rules out 1 pass algorithms. Can we prove that 3 passes are necessary for $\FM{k}$ for some $k\geq 3$? More broadly, it would be very interesting to understand the power of catalytic space for other streaming problems.




\section{Frequency Moment Upper Bounds}\label{sec:upper}

\subsection{\texorpdfstring{Computing $\FM{k}(\stream)$ in $k+1$ pass but less Catalytic Space}{Computing F\_k(S) in k+1 pass}}\label{Subsec: Fk in k+1}


\subsubsection{Facts}
We describe the notion and some results that we use in this subsection. Reader can refer to the well-known book {\em Concrete Mathematics: A Foundation for Computer Science
}~\cite{graham94} for details. 

{\em Stirling number of the second kind}, denoted by \( S(m,r)\), is the number of ways to partition an \(m\)-element set into exactly \(r\) non-empty subsets. The {\em falling factorial}, denoted by  $(j)_r$ is given by, $(j)_r = j(j-1)\dotsm(j-r+1)$. 

\begin{fact}[Stirling numbers of the second kind and falling factorials]\label{Fact: Stirling Numbers}
$S(m,r)$ and  $(j)_r$ satisfy the following properties:
    \begin{multicols}{2}
              \begin{enumerate}
                \item $j^m = \sum_{r = 0}^{m}S(m,r)(j)_r$,
                \item $S(m,m)=1$
                \item $(0)_0=1$
                \item $(0)_r = 0 \qquad \forall r \neq 0$
              \end{enumerate}  
    \end{multicols}
\end{fact}

\begin{fact}[Forward finite differences]\label{Fact: Forward Finite Differences}
    Let $\Delta^1 f(x) = f(x+1)-f(x)$ be the finite difference of a function $f$ at $x$. Let $\Delta^i f(x) = \Delta(\Delta^{i-1}f(x))$. Then we have, $\Delta^k f(0) = \sum_{j=0}^{k}(-1)^{k-j}\binom{k}{j}f(j)$
\end{fact}

\begin{fact}[Finite difference of the falling factorial]\label{Fact: Finite Difference Falling Factorial}
    \begin{align*}
                \Delta^k(0)_r=\begin{cases}
                    0 \text{ if } r \neq k\\
                    r! \text{ if } r = k\\
                \end{cases}
    \end{align*}
\end{fact}

\begin{proof}
    Consider
              \begin{align*}
                \begin{split}
                    \Delta(j)_r &= (j+1)(j)\dotsm(j-r+2) - (j)(j-1)\dotsm(j-r+1)\\
                    &= (j)(j-1)\dotsm(j-r+2)(j+1-(j-r+1))\\
                    &= (j)_{r-1}r
                \end{split}
              \end{align*}
              Therefore, $$\Delta^k(j)_r = r(r-1)\dotsm(r-k+1)(j)_{r-k}$$
              Then,\\
              If $r > k$, we will be left with $(0)_{r-k} = 0$, so $\Delta^k(0)_r = 0$\\
              If $r = k$, and because $(0)_0 = 1$, we have $\Delta^k(0)_r = r!$\\
              If $r < k$, one of the falling factors will be of the form $(r-k+c)=0$ so $\Delta^k(0)_r = 0$
\end{proof}

\subsubsection{Algorithm}

The algorithm to compute $\FM{k}(\stream)$ for any $k \geq 2$, with $k+1$ passes is described in Algorithm~\ref{alg:FkinKP+1}. The algorithms uses $n$ catalytic registers  $\C{C}[i]$ for $i\in [n]$ where each $\C{C}[i]$ can hold $\log M = \lceil k\log_2 mk\rceil$ bits.  




    





\begin{algorithm}
\caption{Stream Power Sum Algorithm}
\label{alg:FkinKP+1}
\begin{algorithmic}[1]
\Require stream $\stream = \langle x_1, \dots, x_{\streamlen}\rangle$, catalytic space $\C{C}$

\State $M \gets 2^{\lceil\log \streamlen^kk!+1\rceil}$
\State $\texttt{sum} \gets (-1)^k\sum_i  (\C{C}[i])^k \mod M$ \Comment{\color{blue} \texttt{sum} =  $(-1)^k\sum_i (\cc{c}[i])^k$} \color{black}

\ 

\For{\textbf{Pass}\ \ $j \gets 1$ to $k$} 

\ 

  \For{$ i \gets 1$ to $\streamlen$}    
       \State $\C{C}[x_i] \gets (\C{C}[x_i] +  1) \mod M$ \Comment{\color{blue} $\C{C}[i] = \cc{c}[i] + j\f_i$} \color{black}
  \EndFor

 \State $\texttt{sum} \gets \texttt{sum} + \sum_i (-1)^{j+k} \binom{k}{j} \cdot (\C{C}[i])^k \mod M$ 
 \
 
 \Comment{\color{blue} \texttt{sum} =  $\sum_i \sum_{\ell = 0}^j(-1)^{\ell+k} \binom{k}{\ell} \cdot (\cc{c}[i] + \ell\f_i)^k$}\color{black}
\EndFor


\textbf{/* Pass $(k+1)$ */}

\ 

\For{$ i \gets 1$ to $\streamlen$}    
        \State $\C{C}[x_i] \gets (\C{C}[x_i] - k) \mod M$ \Comment{\color{blue}{$\C{C}[i] = \cc{c}[i]$}}\color{black}
\EndFor

\ 

\State \Return $\frac{\texttt{sum}}{k!}$ 


\end{algorithmic}
\end{algorithm}

The following claim, which is easy to see, captures the basic behaviour of the algorithm. 

\begin{claim}\label{Claim: FkinK+1 States}
    For any $j \leq k$, at the end of the $j$-th pass, the catalytic space contains
    \[
    \C{C}[i] =( \cc{c}[i] + j\f_{i}) \mod M
    \]
    At the end of the $k+1$-th pass, the catalytic space contains
    \[
    \C{C}[i] = \cc{c}[i] 
    \]
    Furthermore, at the end of $j$-th pass for any $j \leq k$, the counter $\texttt{sum}$ contains 
    \[
    \texttt{sum} =  \sum_i \sum_{\ell = 0}^j(-1)^{\ell+k} \binom{k}{\ell} \cdot ((\cc{c}[i] + \ell\f_i) \mod M)^k \mod M
    \]
\end{claim}


\begin{lemma}\label{Lem: Freq Moment Binomial Identity}
    For any $\cc{c}$ and $\f$, and $k \in \Nat$, we have 
    \[
        \sum_{j=0}^{k}(-1)^j\binom{k}{j}(\cc{c}+j\f)^k = \f^k(-1)^kk!
    \]
\end{lemma}
   
    \begin{proof}
        Note, $(\cc{c}+j\f)^k = \sum_{m=0}^{k}\binom{k}{m}\cc{c}^{k-m}j^m\f^m$, so we have by substitution and reordering summations,
        \begin{align*}
            \sum_{j=0}^{k}(-1)^j\binom{k}{j}(\cc{c}+j\f_i)^k &= \sum_{j=0}^{k}(-1)^j\binom{k}{j}\sum_{m=0}^{k}\binom{k}{m}\cc{c}^{k-m}j^m\f^m\\
            &= \sum_{m=0}^{k}\binom{k}{m}\cc{c}^{k-m}\f^m\sum_{j=0}^{k}(-1)^j\binom{k}{j}j^m\\
            &= \sum_{m=0}^{k}\binom{k}{m}\cc{c}^{k-m}\f^m\sum_{j=0}^{k}(-1)^j\binom{k}{j}\sum_{r = 0}^{m}S(m,r)(j)_r\\
            &= \sum_{m=0}^{k}\sum_{r = 0}^{m}\binom{k}{m}\cc{c}^{k-m}\f^mS(m,r)\sum_{j=0}^{k}(-1)^j\binom{k}{j}(j)_r
        \end{align*}
        By \Cref{Fact: Forward Finite Differences}, we have that 
        \[
        \sum_{j=0}^{k}(-1)^j\binom{k}{j}(j)_r = (-1)^k\Delta^k(0)_r
        \]
        Thus, from \Cref{Fact: Stirling Numbers}, we have
        \[
        \sum_{m=0}^{k}\sum_{r = 0}^{m}\binom{k}{m}\cc{c}^{k-m}\f^mS(m,r)(-1)^k\Delta^k(0)_r
        \]
        By \Cref{Fact: Finite Difference Falling Factorial}, we have that $\Delta^k(0)_r \neq 0$ only when $r = k$. Therefore, all terms are 0 except when
        $m=k$ and $r=k$ ($m$ must be $k$, as the bound on $r$ limits it to be at most $m$). Therefore we have
        \begin{align*}
                &\sum_{m=0}^{k}\sum_{r = 0}^{m}\binom{k}{m}\cc{c}^{k-m}\f^mS(m,r)(-1)^k\Delta^k(0)_r\\
                = &\binom{k}{k}c^{k-k}\f^kS(k,k)(-1)^k\Delta^k(0)_k\\
                = &\f^k(-1)^kk!
        \end{align*}
    \end{proof}

    \begin{corollary}\label{coro:identity}
       For any $i\in [n]$,  $\sum_{\ell = 0}^j (-1)^{\ell+k} \binom{k}{\ell} \cdot ((\cc{c}[i] + \ell\f_i) \mod M)^k = ((\f_i)^kk!) \mod M$
    \end{corollary}


\begin{theorem}\label{thm:kplus1}
    \Cref{alg:FkinKP+1} takes as input a stream $\stream$ uses $k+1$ passes, $\bigo{k\log \streamlen k}$ clean space, $O(kn\log mk)$ catalytic space and outputs the value of $\FM{k}(\stream)$.
    
\end{theorem}

\begin{proof}
    By \Cref{Claim: FkinK+1 States}, we have that at the end of the $k$-th pass,
    \begin{align*}
        \texttt{sum} = &(\sum_i \sum_{\ell = 0}^k (-1)^{\ell+k} \binom{k}{\ell} \cdot ((\cc{c}[i] + \ell\f_i)\mod M)^k) \mod M&\\
        = & \left(\sum_{i} \f_i^kk!\right)\mod M & \mbox{[From Corollary~\ref{coro:identity}]}\\
        = & k!\sum_{i} \f_i^k & \mbox{[Since $M \geq k!\sum_{i} \f_i^k$.]}
    \end{align*}
    Hence, we have $\frac{\texttt{sum}}{k!} = \sum_i \f_i^k = \FM{k}(\stream)$, which the algorithm returns. Note that the only non-catalytic space that is needed is to store the variables $\texttt{sum}$, $i$, and $j$. Note that the maximum value  $\texttt{sum}$ can take is at most $M$, leading to a clean space bound of $O(k \log mk)$. For the catalytic space, observe that each element $i \in [n]$ corresponds to a single catalytic register that has size at most $M$, i.e. uses $O(k\log mk)$ bits.
\end{proof}

\vspace{2mm}    
\noindent{\em Remark:}The \(k+1\)-pass algorithm is stronger than a standard multi-pass streaming algorithm in the following sense. In the traditional multi-pass model, every pass sees the stream in the same order. In contrast, the algorithm for \(\FM{k}\) works even if, on each pass, an adversary presents the elements in an arbitrary order, as long as the underlying stream itself is unchanged.

\subsubsection{\texorpdfstring{Computing $\FM{2}$ in $2$ passes}{Computing F\_2 in 2 passes}}\label{Section: 2-pass f_2}

The high level discussion of recursive decomposition based algorithm for $\FM{2}$ computation assumes that the tasks of computing $\FM{2}(\stream_L)$, $\FM{2}(\stream_R)$, and $\sum_i g_i h_i$ can be carried out in parallel within two passes using catalytic space. We now formalize this via a non-recursive version of the algorithm. 

We use $ \lceil  \log \streamlen \rceil$ catalytic arrays and run $ \lceil \log \streamlen \rceil$ parallel algorithms, where the $j$th algorithm uses the $j$th catalytic array. In the $j$th iteration, we partition the stream $\stream$ into substreams of size $2^j$. Denote these substreams by
$\stream^{(j,0)}, \stream^{(j,1)}, \dots, \stream^{\left(j,\left\lfloor \frac{\streamlen - 1}{2^j} \right\rfloor\right)}$.
Let their corresponding frequency vectors be 
$\mathbf{g}^{(j,0)}, \mathbf{h}^{(j,0)}, \mathbf{g}^{(j,2)}, \mathbf{h}^{(j,2)}, \dots $. I.e.
$\mathbf{g}^{(j,0)}$ is the frequency vector $\stream^{(j,0)}$ and $\mathbf{h}^{(j,0)}$ is the frequency vector of $\stream^{(j,1)}$ and so on. The catalytic space in the $j$th algorithm is used to compute
\[
\sum_{\substack{k \\ k \text{ even}}} \sum_i \g^{(j,k)}_i \, \h^{(j,k)}_i.
\]

Finally, summing the outputs over all $\lceil \log \streamlen \rceil$ parallel algorithms yields
\[
\sum_{j=0}^{\lceil \log \streamlen \rceil} \;\sum_{\substack{k \\ k \text{ even}}} \;\sum_i \g^{(j,k)}_i \, \h^{(j,k)}_i,
\]
which is exactly equal to $\FM{2}(\stream) - \streamlen$.

Now, we present the Algorithm~\ref{alg:F2in2pass} and we establish its performance in Theorem~\ref{Thm: F2in2}.

\begin{algorithm}[H]
\caption{$\FM{2}$-in-2-pass-Algorithm}
\label{alg:F2in2pass}
\begin{algorithmic}[1]
\Require The stream $\stream = \langle x_1, \dots, x_{\streamlen}\rangle$, Catalytic space $\C{C}$

\State $M \gets \streamlen^2$

\textbf{/* Pass 1 */}

\For{$i \gets 1$ to $\streamlen$}
   \For{$j \gets 0$ to $\lceil \log \streamlen\rceil -1$}
    \State $k \gets \lfloor (i-1)/2^j\rfloor$
    \If{$k$ is even}
      \State $\C{C}^{(j,k)}[x_i] = (\C{C}^{(j,k)}[x_i]+1) \mod M$\Comment{\blue{$ \C{C}[i] = \cc{c}[i] + \g_i$}}
    \EndIf
     \If{$k$ is odd}
      \State $\W{S}^{(j)} = (\W{S}^{(j)} + \C{C}^{(j,k-1)}[x_i]) \mod M$\Comment{\blue{$ \W{S} = \sum_i (\cc{c}[i] + \g_i)\h_i$}}
    \EndIf
    \EndFor
\EndFor

\textbf{/* Pass 2 */}

\For{$i \gets 1$ to $\streamlen$}
   \For{$j \gets 0$ to $\lceil \log\streamlen \rceil -1$}
    \State $k \gets \lfloor (i-1)/2^j\rfloor$
    \If{$k$ is even}
      \State $\C{C}^{(j,k)}[x_i] = (\C{C}^{(j,k)}[x_i]-1) \mod M$\Comment{\blue{$\C{C}[i] = \cc{c}[i]$}}
    \EndIf
     \If{$k$ is odd}
      \State $\W{S}^{(j)} = (\W{S}^{(j)} - \C{C}^{(j,k-1)}[x_i])  \mod M$\Comment{\blue{$\W{S} = \sum_i \g_i\h_i$}}
    \EndIf
    \EndFor
\EndFor


\State \Return $\left(\streamlen + 2\sum_{j=1}^{\lceil\log \streamlen\rceil}\W{S}^{(j)} \right)$

\end{algorithmic}
\end{algorithm}


\begin{theorem}\label{Thm: F2in2}
    \Cref{alg:F2in2pass} takes as input a stream $\stream$ uses $2$ passes, $\bigo{\log m}$ clean space, $\bigo{n\log m}$ catalytic space, and outputs the value of $\FM{2}(\stream)$.
\end{theorem}

Before we give the proof of the algorithm we will need a few notations. 
Let the stream $\stream$ be $\langle x_1, \dots x_{\streamlen}\rangle$ with $x_i \in [n]$. For any $a,b\in \mathbb{N}$ if $1\leq a \leq b\leq \streamlen$, we will denote by $\f\langle [a, b]\rangle$ the frequency vector of the sub-stream $\langle x_a, \dots x_b\rangle$. That is,  for all $i \in [1, \dots, n]$
\begin{equation}\label{eq:FGH}
\f\langle [a, b]\rangle_i = \left| \left\{a\leq j \leq b \mid x_j = i  \right\}\right|
\end{equation}

For a fixed $j$ consider the intervals, $[1 + k\cdot 2^j, 2^{j+1} + k\cdot 2^j]$, as $k$ ranges over all even values from $0$ to $\lfloor (\streamlen -1)/2^j\rfloor$. These intervals partitions the stream into $\ceil{\lfloor (\streamlen -1)/2^j\rfloor/2}$ disjoint intervals of length $2^{j+1}$ each. 
For any $j\in [1, \dots, \lceil\log \streamlen\rceil]$ and any $k \in \{0, \dots, \lfloor(\streamlen-1)/2^j\rfloor\}$ let $\g^{(j,k)}$ and  $\h^{(j,k)}$ denote the frequency vector $\f\langle  [1+ k\cdot2^{j} , 2^{j} + k\cdot2^{j} ] \rangle$ and $\f\langle  [2^{j} +1 + k\cdot2^{j} , 2^{j+1} + k\cdot2^{j} ] \rangle$ respectively. 
In other words, if we consider the interval $[1 + k\cdot 2^j, 2^{j+1} + k\cdot 2^j]$, 
$\g^{(j,k)}$ is the frequency vector of the first half of the interval and $\h^{(j,k)}$ is the frequency vector of the second half of the interval.  Also note that, 
$$\f\langle  [1 + k\cdot2^{j}, 2^{j+1} + k\cdot2^{j}] \rangle = \g^{(j,k)} + \h^{(j,k)}$$
The following two claims can be verified from the algorithm. We denote by $c^{(j,k)}[i]$ the content in the catalytic memory block $C^{(j,k)}[i]$ at the start of the algorithm.

\begin{claim}
    At the end of the first pass for all  $j\in [1, \dots, \lceil\log \streamlen\rceil]$, for all  even $k\in [0, \lfloor (i-1)/2^j\rfloor]$ and for all
    $i\in [1, \dots, n]$,
   $$ C^{(j,k)}[i] = \left(c^{(j,k)}[i] + \left(\g^{(j,k)}\right)_i\right)
   \mod M$$
   and, 
   \begin{align*}
   S^{(j)} =\left(\sum_{\substack{k= 0\\ k \mbox{ even }}}^{\lfloor(\streamlen-1)/2^j\rfloor}\sum_{i=1}^{n} \left(c^{(j,k)}[i] + \left(\g^{(j,k)}\right)_i\right)\cdot \left(\h^{(j,k)}\right)_i\right) \mod M
    \end{align*}
\end{claim}


\begin{claim}\label{cl:2pass2}
    At the end of both the passes for all  $j\in [1, \dots, \lceil\log \streamlen\rceil]$, for all  even $k\in [0, \lfloor (i-1)/2^j\rfloor]$ and for all
    $i\in [1, \dots, n]$,
   $$ C^{(j,k)}[i] = c^{(j,k)}[i] $$ 
    and 
    \begin{align*}
   S^{(j)} =\left(\sum_{\substack{k= 0\\ k \mbox{ even }}}^{\lfloor(\streamlen-1)/2^j\rfloor}\sum_{i=1}^{n} \left(\g^{(j,k)}_i\right)\cdot \left(\h^{(j,k)}\right)_i\right) \mod M
    \end{align*}
\end{claim}


\begin{lemma}\label{lem:F2in2pass}
  At the end of both the pass for all $j\in [1, \dots, \lceil\log \streamlen\rceil]$, $$\streamlen + \sum_{i=0}^{j-1} S^{(i)} = 
   \sum_{\substack{k= 0\\ k \mbox{ even }}}^{\lfloor(\streamlen-1)/2^j\rfloor}
  \left\Vert\f\langle[1 + k\cdot2^{j}, 2^{j+1} + k\cdot2^{j}] \rangle
  \right\Vert^2_2$$
\end{lemma}

\begin{proof}
    The proof follows from induction on $j$.  From Equation~\ref{eq:FGH},
$\f\langle [a, b]\rangle_i = \left| \left\{a\leq j \leq b \mid x_j = i  \right\}\right|$, hence $$\left\Vert\f\langle[1 + k\cdot2^{j}, 2^{j+1} + k\cdot2^{j}] \rangle
  \right\Vert^2_2 = \left\Vert\g^{(j,k)}\right\Vert^2_2 + \left\Vert\h^{(j,k)}\right\Vert^2_2 + 2\cdot\g^{(j,k)}\cdot\h^{(j,k)}.$$
  The rest of the correctness-proof follows by induction on $j$ and Claim~\ref{cl:2pass2}. The algorithm is clearly a 2-pass algorithm. For the space complexity, each $\W{S}^{(j)}$ requires $O(\log m)$ bits. Since $j$ ranges from $0$ to $\lceil \log m\rceil -1$, the total space is $O(\log^2 m)$. However, in the end (line 16), for the final output we only need the sum of all the $\W{S}^{(j)}$'s. Thus, instead of storing them separately, we can keep track of a global sum, leading to $O(\log m)$ space algorithm. The rest of the correctness-proof follows by induction on j and \Cref{cl:2pass2}.
   
\end{proof}

\begin{proof}[Proof of~\Cref{Thm: F2in2}]
    The correctness of the algorithm follows from \Cref{lem:F2in2pass}. The algorithm is clearly a 2-pass algorithm. For the space complexity, each $\W{S}^{(j)}$ requires $O(\log m)$ bits. Since $j$ ranges from $0$ to $\lceil \log m\rceil -1$, the total space is $O(\log^2 m)$. However, in the end (line 16), for the final output we only need the sum of all the $\W{S}^{(j)}$'s. Thus, instead of storing them separately, we can keep track of a global sum, leading to $O(\log m)$ space algorithm. For the catalytic space, observe that we use $O(1)$ registers for each element $i \in [n]$, and each register uses at most $O(\log m)$ bits.
\end{proof}

\subsubsection{\texorpdfstring{Computing $\FM{3}$ in $3$ passes}{Computing F\_3 in 3 passes}}\label{Section: F3 in 3passes}

We present the 3-pass catalytic algorithm for computing $\FM{3}$ in Algorithm~\ref{alg:F3in3pass}.
The correctness of the algorithm is presented in Theorem~\ref{Thm: F3in3}, and it follows from Lemma~\ref{lem:F3in3pass}. The three claims, Claim~\ref{cl:3pass1}, \ref{cl:2pass2} and \ref{cl:3pass3}, can be easily verified from the algorithm.

\begin{algorithm}
\caption{$\FM{3}$-in-3-pass-Algorithm}
\label{alg:F3in3pass}
\begin{algorithmic}[1]
\Require The stream $\stream = \langle x_1, \dots, x_{\streamlen}\rangle$, Catalytic space $\C{A}$, $\C{B}$, $\C{C}$ and $\C{D}$

\State $M \gets \streamlen^3$
\State Initialize for all $(i,j)$, $\W{S}^{(j)}, \W{T}^{(j)}, \W{U}^{(j)}, \W{V}^{(j)}, \W{W}^{(j)}, \W{X}^{(j)}, \W{Y}^{(j)}, \W{Z}^{(j)} \gets 0$

\textbf{/* Pass 1 */}

\For{$i \gets 1$ to $\streamlen$}
   \For{$j \gets 0$ to $\lceil \log \streamlen\rceil -1$}
    \State $k \gets \lfloor (i-1)/2^j\rfloor$
    \If{$k$ is even}
      \State $\C{C}^{(j,k)}[x_i] = (\C{C}^{(j,k)}[x_i]+1) \mod M$ \Comment{\blue{$ \C{C}[i] = \cc{c}[i] + \g_i$}}
      \State $\C{B}^{(j,k)}[x_i] = (\C{B}^{(j,k)}[x_i]+ \C{A}^{(j,k)}[x_i] )
      \mod M$
      \Comment{\blue{$\C{B}[i] = \cc{b}[i] + \cc{a}[i]\g_i$}}
      \State $\W{S}^{(j)} = (\W{S}^{(j)} + \C{A}^{(j,k)}[x_i]^2 ) \mod M$ \Comment{\blue{$\W{S} = \sum_i \cc{a}[i]^2 \g_i$}}
      \State $\W{U}^{(j)} = (\W{U}^{(j)} + \C{D}^{(j,k)}[x_i]^2) \mod M$
      \Comment{\blue{$\W{U} = \sum_i \cc{d}[i]^2 \g_i$}}
    \EndIf
     \If{$k$ is odd}
      \State $\C{A}^{(j,k)}[x_i] = (\C{A}^{(j,k)}[x_i]+1) \mod M$
      \Comment{\blue{$ \C{A}[i] = \cc{a}[i] + \h_i$}}
      \State $\W{X}^{(j)} = (\W{X}^{(j)} + \C{C}^{(j,k-1)}[x_i]^2 ) \mod M$
      \Comment{\blue{$ \W{X} = \sum_i (\cc{c}[i] + \g_i)^2\h_i$}}
      \State $\W{Y}^{(j)} = (\W{Y}^{(j)} + \C{B}^{(j,k-1)}[x_i]) \mod M$
      \Comment{\blue{$ \W{Y} = \sum_i (\cc{b}[i] + \cc{a}[i]\g_i)\h_i$}}
    \EndIf
    \EndFor
\EndFor

\textbf{/* Pass 2 */}

\For{$i \gets 1$ to $\streamlen$}
   \For{$j \gets 0$ to $\lceil \log\streamlen \rceil-1$}
    \State $k \gets \lfloor (i-1)/2^j\rfloor$
    \If{$k$ is even}
      \State $\C{C}^{(j,k)}[x_i] = (\C{C}^{(j,k)}[x_i]-1) \mod M$
      \Comment{\blue{$ \C{C}[i] = \cc{c}[i]$}}
      \State $\W{Z}^{(j)} = (\W{Z}^{(j)} + \C{A}^{(j,k-1)}[x_i]^2) \mod M$
      \Comment{\blue{$ \W{Z}  = \sum_i (\cc{a}[i] + \h_i)^2\g_i$}}
    \EndIf
     \If{$k$ is odd}
        \State $\C{D}^{(j,k)}[x_i] = (\C{D}^{(j,k)}[x_i] + \C{C}^{(j,k)}[x_i]) \mod M$
        \Comment{\blue{$ \C{D}[i] = (\cc{d}[i] + \cc{c}[i])\h_i$}}
        \State $\C{A}^{(j,k)}[x_i] = (\C{A}^{(j,k)}[x_i] - 1) \mod M$
        \Comment{\blue{$ \C{A}[i] = \cc{a}[i]$}}
        \State $\W{W}^{(j)} = (\W{W}^{(j)} + \C{C}^{(j,k-1)}[x_i]^2) \mod M$
        \Comment{\blue{$ \W{W} = \sum_i (\cc{c}[i])^2 \h_i$}}
    \EndIf
    \EndFor
\EndFor

\textbf{/* Pass 3 */}

\For{$i \gets 1$ to $\streamlen$}
   \For{$j \gets 0$ to $\lceil \log\streamlen \rceil-1$}
    \State $k \gets \lfloor (i-1)/2^j\rfloor$
    \If{$k$ is even}
      \State $\C{B}^{(j,k)}[x_i] = (\C{B}^{(j,k)}[x_i]- \C{A}^{(j,k)}[x_i]) \mod M$
      \Comment{\blue{$ \C{B}[i] = \cc{b}[i]$}
}
      \State $\W{T}^{(j)} = (\W{T}^{(j)} + \C{D}^{(j,k-1)}[x_i]) \mod M$
      \Comment{\blue{$ \W{T} = \sum_i (\cc{d}[i] + \cc{c}[i])\h_i\g_i$}
}
    \EndIf
     \If{$k$ is odd}
        \State $\C{D}^{(j,k)}[x_i] = (\C{D}^{(j,k)}[x_i] - \C{C}^{(j,k)}[x_i]) \mod M$
        \Comment{\blue{$ \C{D}[i] = \cc{d}[i]$}
}
        \State $\W{V}^{(j)} = (\W{V}^{(j)} + \C{B}^{(j,k-1)}[x_i]) \mod M$
        \Comment{\blue{$ \W{V} = \sum_i \cc{b}[i]\h_i$}
}
    \EndIf
    \EndFor
\EndFor

\ 


\State $\W{P} = \sum_{j=0}^{\lceil\log \streamlen\rceil-1}\left(\W{X}^{(j)} - \W{W}^{(j)} - 2(\W{T}^{(j)} - \W{U}^{(j)})\right) \mod M$

\State $\W{Q} =  \sum_{j=0}^{\lceil\log \streamlen\rceil-1}\left(\W{Z}^{(j)} - \W{S}^{(j)} - 2(\W{Y}^{(j)} - \W{V}^{(j)})\right) \mod M$

\State \Return $(m + 3\W{P} + 3\W{Q})$

\end{algorithmic}
\end{algorithm}

\begin{claim}\label{cl:3pass1}
    At the end of the first pass for all  $j\in [1, \dots, \lceil\log \streamlen\rceil]$, for all
    $i\in [1, \dots, n]$ and for all even $k\in [0, \lfloor (i-1)/2^j\rfloor]$,
   $$ \C{C}^{(j,k)}[i] = \left(\cc{c}^{(j,k)}[i] + \left(\g^{(j,k)}\right)_i\right)
   \mod M$$
   and,
   $$ \C{B}^{(j,k)}[i] = \left(\cc{b}^{(j,k)}[i] + \left(\g^{(j,k)}\right)_i\fbrac{\cc{a}^{(j,k)}}_i\right)
   \mod M$$
    And and for all odd $k\in [0, \lfloor (i-1)/2^j\rfloor]$,
    $$ \C{A}^{(j,k)}[i] = \left(\cc{a}^{(j,k)}[i] + \left(\h^{(j,k)}\right)_i\right)
   \mod M$$
    Furthermore, we have
   \begin{align*}
   \W{S}^{(j)} =\left(\sum_{\substack{k= 0\\ k \mbox{ even }}}^{\lfloor(\streamlen-1)/2^j\rfloor}\sum_{i=1}^{n} \left(\fbrac{\cc{a}^{(j,k)}}_i^2 \left(\g^{(j,k)}\right)_i\right)\right) \mod M
    \end{align*}
    \begin{align*}
   \W{U}^{(j)} =\left(\sum_{\substack{k= 0\\ k \mbox{ even }}}^{\lfloor(\streamlen-1)/2^j\rfloor}\sum_{i=1}^{n} \left(\fbrac{\cc{d}^{(j,k)}}_i^2 \left(\g^{(j,k)}\right)_i\right)\right) \mod M
    \end{align*}
    \begin{align*}
   \W{X}^{(j)} =\left(\sum_{\substack{k= 0\\ k \mbox{ even }}}^{\lfloor(\streamlen-1)/2^j\rfloor}\sum_{i=1}^{n} \left(\fbrac{\fbrac{\cc{d}^{(j,k)}}_i+\left(\g^{(j,k)}\right)_i}^2 \left(\h^{(j,k)}\right)_i\right)\right) \mod M
    \end{align*}
    \begin{align*}
   \W{Y}^{(j)} =\left(\sum_{\substack{k= 0\\ k \mbox{ even }}}^{\lfloor(\streamlen-1)/2^j\rfloor}\sum_{i=1}^{n} \left(\fbrac{\fbrac{\cc{b}^{(j,k)}}_i+\left(\cc{a}^{(j,k)}\right)_i\left(\g^{(j,k)}\right)_i} \left(\h^{(j,k)}\right)_i\right)\right) \mod M
    \end{align*}
\end{claim}

\begin{claim}\label{cl:3pass2}
    At the end of the second pass for all  $j\in [1, \dots, \lceil\log \streamlen\rceil]$, for all $i\in [1, \dots, n]$ and for all even $k\in [0, \lfloor (i-1)/2^j\rfloor]$,
    $$ \C{C}^{(j,k)}[i] = \cc{c}^{(j,k)}[i]
   \mod M$$
    and for all odd $k\in [0, \lfloor (i-1)/2^j\rfloor]$,
    $$ \C{D}^{(j,k)}[i] = \left(\fbrac{\cc{d}^{(j,k)}[i]+\cc{c}^{(j,k)}[i]}\h^{(j,k)}_i\right)
   \mod M$$
   $$ \C{A}^{(j,k)}[i] = \cc{a}^{(j,k)}[i]
   \mod M$$
   Furthermore, we have
   \begin{align*}
   \W{Z}^{(j)} =\left(\sum_{\substack{k= 0\\ k \mbox{ odd }}}^{\lfloor(\streamlen-1)/2^j\rfloor}\sum_{i=1}^{n} \left(\fbrac{\cc{a}^{(j,k)}}_i \left(\h^{(j,k)}\right)_i\right)^2\h_i^{(j,k)}\right) \mod M
    \end{align*}
    \begin{align*}
   \W{W}^{(j)} =\left(\sum_{\substack{k= 0\\ k \mbox{ even }}}^{\lfloor(\streamlen-1)/2^j\rfloor}\sum_{i=1}^{n} \fbrac{\cc{c}^{(j,k)}}_i^2\h_i^{(j,k)}\right) \mod M
    \end{align*}
\end{claim}

\begin{claim}\label{cl:3pass3}
    At the end of the second pass for all  $j\in [1, \dots, \lceil\log \streamlen\rceil]$, for all $i\in [1, \dots, n]$ and for all even $k\in [0, \lfloor (i-1)/2^j\rfloor]$,
    $$ \C{B}^{(j,k)}[i] = \cc{b}^{(j,k)}[i]
   \mod M$$
    and for all odd $k\in [0, \lfloor (i-1)/2^j\rfloor]$,
    $$ \C{D}^{(j,k)}[i] = \cc{d}^{(j,k)}[i]
   \mod M$$
   Furthermore, we have
   \begin{align*}
   \W{T}^{(j)} =\left(\sum_{\substack{k= 0\\ k \mbox{ odd }}}^{\lfloor(\streamlen-1)/2^j\rfloor}\sum_{i=1}^{n} \left(\fbrac{\cc{d}^{(j,k)}}_i +\fbrac{\cc{c}^{(j,k)}}_i\right)\g_i^{(j,k)}\h_i^{(j,k)}\right) \mod M
    \end{align*}
    \begin{align*}
   \W{V}^{(j)} =\left(\sum_{\substack{k= 0\\ k \mbox{ even }}}^{\lfloor(\streamlen-1)/2^j\rfloor}\sum_{i=1}^{n} \fbrac{\cc{b}^{(j,k)}}_i\h_i^{(j,k)}\right) \mod M
    \end{align*}
\end{claim}

At the end of three passes let $$\W{P}^{(j)} = \sum_{\ell=0}^{j-1}\left(\W{X}^{(\ell)} - \W{W}^{(\ell)} - 2(\W{T}^{(\ell)} - \W{U}^{(\ell)})\right) \mod M$$ and let

$$\W{Q}^{(j)} =  \sum_{\ell=0}^{j-1}\left(\W{Z}^{(\ell)} - \W{S}^{(\ell)} - 2(\W{Y}^{(\ell)} - \W{V}^{(\ell)})\right) \mod M$$

\begin{lemma}\label{lem:F3in3pass}
  At the end of three passes for all $j\in [1, \dots, \lceil\log \streamlen\rceil]$, $$\streamlen + 3\W{P}^{(j)} + 3\W{Q}^{(j)}= 
   \sum_{\substack{k= 0\\ k \mbox{ even }}}^{\lfloor(\streamlen-1)/2^j\rfloor}
  \left\Vert\mathsf{F}\langle[1 + k\cdot2^{j}, 2^{j+1} + k\cdot2^{j}] \rangle
  \right\Vert^3_3$$
\end{lemma}

\begin{proof}
\begin{align*}
\sum_{i} (\f^{(j,k)}_i)^3 & = \sum_i \fbrac{(\g^{(j,k)}_i+\h^{(j,k)}_i)^3}\\ 
& = \sum_i (\g^{(j,k)}_i)^3 + \sum_i (\h^{(j,k)}_i)^3 + 3 \sum_i (\g^{(j,k)}_i)^2\h^{(j,k)}_i + 3 \sum_i \g^{(j,k)}_i(\h^{(j,k)}_i)^2
\end{align*}
\end{proof}

Combining the above lemma's, we obtain the theorem.

\begin{theorem}\label{Thm: F3in3}
    \Cref{alg:F3in3pass} takes as input a stream $\stream$ uses $3$ passes, $\bigo{\log m}$ clean space, $\bigo{n \log m}$ catalytic space, and outputs the value of $\FM{3}(\stream)$.
\end{theorem}

\begin{proof}
        The correctness of the algorithm follows from \Cref{lem:F3in3pass}. The algorithm is clearly a 3-pass algorithm. Observe that each $\W{S}^{(j)}$, $\W{T}^{(j)}$, $\W{U}^{(j)}$, $\W{V}^{(j)}$, $\W{W}^{(j)}$, $\W{X}^{(j)}$, $\W{Y}^{(j)}$, $\W{Z}^{(j)}$ requires $O(\log M) = O(\log m)$ bits of space. Since $j$ ranges from $0$ to $\lceil \log m\rceil -1$, the total space is $O(\log^2 m)$. However, in the end (line 34-36), for the final output we only need the sum of all the $\W{S}^{(j)}$'s. Thus, instead of storing them separately, we can keep track of a global sum, leading to $O(\log m)$ space algorithm. For the catalytic space, observe that we use $O(1)$ registers for each element $i \in [n]$, and each register uses at most $O(\log m)$ bits.
\end{proof}



\subsection{\texorpdfstring{Computing $\FM{0}$ for $t$-flat streams }{Computing F\_0 for t-flat streams}}\label{Sec: Tflat}


We first formally define the $t$-flat streams.

\begin{definition}
    A stream $\stream$ is $t$-flat if for every $i \in [n]$, $\f_i \leq t$.
\end{definition}

We use the following fact about prime numbers for our proofs.

\begin{lemma}
Let \(1 < t < m\) be integers, and let
$
\ell = \left\lceil \frac{\log n}{\log t} \right\rceil
$. 
Then there exist distinct primes
$t < p_1 < p_2 < \cdots < p_\ell$
such that
$
\prod_{i=1}^\ell p_i > n.
$
Moreover, if \(p_\ell\) denotes the largest of these primes, then
$
p_\ell = O(t + \frac{\ln n\ln\ln n }{ \ln t})$. 
\end{lemma}

\begin{proof}
Let \(p_1 < p_2 < \cdots < p_\ell\) be the first \(\ell\) primes exceeding \(t\). Since each
\(p_i > t\), we have
$
\prod_{i=1}^\ell p_i > t^\ell.
$
By the choice of \(\ell\),
$
\ell \ge \frac{\log n}{\log t},
$
and hence
$
t^\ell \ge n.
$
Therefore,
$
\prod_{i=1}^\ell p_i > n.
$. 

It remains to bound \(p_\ell\). Let \(\pi(t)\) be the number of primes less than or equal to \(t\).
The \(\ell\)-th prime greater than \(t\) is the \((\pi(t)+\ell)\)-th prime in the sequence of all primes. By prime number theorem $\pi(t) \sim t/\ln t$ and $p_l \sim (\pi(t)+\ell)\ln (\pi(t)+\ell)$. 
This can be estimated to be $O(t + \frac{\ln n\ln\ln n }{ \ln t})$. 
\end{proof}

Now, we state our result for computing $\FM{0}$ for $t$-flat streams.

\begin{theorem}\label{Theorem: F0 for tflats}
Let \(1 < t < m\) be integers, and let
$t < p_1 < p_2 < \cdots < p_\ell$ where $p_i$s are prime number
such that
$
\prod_{i=1}^\ell p_i > n.
$
    There is a deterministic, $p_\ell$-pass, catalytic space algorithm that observes a $t$-flat stream $\mathcal{S}$ and returns $\FM{0}(\mathcal{S})$. The algorithm uses ${O}(\ell \cdot p_\ell \cdot\log mp_\ell + \log n)$ space. 
\end{theorem}

For the proof of \Cref{Theorem: F0 for tflats}, we require the following claim.

\begin{claim}\label{Claim: Fermats for tflats}
    For a $t$-flat stream, and a prime $p \geq t$ , we have $\FM{p-1} \pmod p = \FM{0} \pmod p  $
    \end{claim}
\begin{proof}
Note that, since the stream is $t$-flat, for every $p \in L$, $\f_i < p$. Thus from Fermat's little theorem, $\f_i^{p-1} = 1 \pmod p$ if $\f_{i} \neq 0$ and $0 \pmod p$ if $\f_i= 0$. Thus for every $p \in L$
\[\FM{0} = |\{i\mid \f_i^{p-1} \pmod p = 1\}| \] 
Equivalently 
\[\FM{0} = \sum_{i=1}^n \f_{i}^{p-1} \pmod p\]
where the sum is taken over the integer ring.
 Now,
\begin{eqnarray*}
    \FM{p-1} \pmod p  & = & \left(\sum_{i=1}^n \f_{i}^{p-1} \right) \pmod p\\
    & = & \sum_{i=1}^n \left (\f_{i}^{p-1} \pmod p\right) \pmod p\\
    & = & \FM{0} \pmod p
\end{eqnarray*}
This completes the proof of the claim. 
\end{proof}

Now, we establish \Cref{Theorem: F0 for tflats}.

\begin{proof}[Proof of \Cref{Theorem: F0 for tflats}]
The algorithm works as follows. Let $L = \{p_1,\cdots,p_l\}$ be the list of prime numbers such that $p_i > t$ and $
\prod_{i=1}^\ell p_i > n.
$ For each $p \in L$, compute $\FM{p-1} \pmod p$, in parallel. By~\Cref{{thm:kplus1}}, this can done using $p_\ell$ passes using space $O(\ell\cdot p_\ell\cdot \log mp_\ell)$. 
We have the following claim. 

Since the product of primes in $L$ is bigger than $n$ and $\FM{0} \leq n$, by using Chinese remaindering, we can recover $\FM{0}$. The additional space to do Chinese remaindering is $O(\log n)$. In this case, we can go over all integers $a \leq n$ and find $a$ that satisfies all the modulo equations. 
\end{proof}

We note that even for 2-flat streams, any randomized $k$-pass algorithm for exact computation of $\FM{0}$ 
in the standard streaming model will require $\Omega(n/k)$ space~\cite{RAZBOROV1992385}. 

\vspace{2mm}
\begin{remark} In fact, the above proof can be adapted to more general class of streams. What we require is that there are a set of small prime numbers which do not divide any $\f_i$. Thus, for example, if $\f_i$ are {\em smooth} numbers, numbers for which all its prime divisors are small, the above algorithm will work. $t$-smoothness (all prime factors of $\f_i$ are $\leq t$) is a generalization of $t$-flat, since $\f_i \leq t$, each prime factors of $\f_i$ are also $\leq t$.   
\end{remark}

\section{Counting Subgraphs in 4 Passes}
\label{subsec:subgraphcounting}

In this section, we show how our exact $\FM{k}$ algorithms can be used to count subgraphs in the catalytic-space model. Using our $O(k\log n)$-space algorithm for computing $\FM{k}$, we prove that, for every fixed graph $H$, the number of induced copies of $H$ can be computed exactly in four passes using $O_H(\log n)$ space. As a corollary, by exploiting our specialized $O(\log n)$-space algorithms for computing $\FM{2}$ and $\FM{3}$, we obtain an exact three-pass algorithm for counting triangles using only $O(\log n)$ space. Here, $O_H(\cdot)$ denotes that the hidden constant depends only on the fixed graph $H$.

\begin{theorem}[Exact subgraph counting]\label{Thm: Counting SUbgraphs}
Let $H$ be a fixed graph on $r$ vertices and $q$ edges. Then the number of copies of $H$, as well as the number of induced copies of $H$ in an $n$-vertex graph stream can be computed exactly in four passes using $O_H(\log n)$ space.
In particular, 
\begin{itemize}
    \item Number of copies of $H$ can be counted in four passes using $\bigo{q^2\log n}$ clean space, and $O(q^3 n\log n)$ catalytic space.
    \item Number of induced copies of $H$ can be counted in four passes using $\bigo{2^M M^2q^2\log n}$ clean space, and $O(2^M M^2q^3\log n)$ catalytic space, where $M = \binom{r}{2}$.
\end{itemize}
\end{theorem}

\begin{proof} We use the idea from  Bar-Yossef, Kumar, and Sivakumar~\cite{BarYossefKS2002} who gave a reduction from triangle counting to frequency moments.
We first describe how to count non-induced copies of a fixed graph.
Let $J$ be a fixed graph on vertex set $[r]$, 
$q=|E(J)|$. Let $H$ be a fixed graph on $r$ labelled vertices, and consider an arbitrary input graph $G$ on $n$ vertices. We distinguish two types of counts:

\begin{itemize}
    \item $N(J)$: the number of labelled \emph{non-induced} copies of a graph $J$ in $G$. That is, $N(J)$ is the number of injective mappings $\phi: V(J) \hookrightarrow V(G)$ such that every edge of $J$ is present in $G[\phi(V(J))]$. Extra edges among the image vertices are allowed.
    \item $I(J)$: the number of labelled \emph{induced} copies of $J$ in $G$. Here we require that the image of $\phi$ induces exactly the graph $J$ --- no extra edges beyond those of $J$ are permitted.
\end{itemize}

\paragraph{Counting $N(J)$:}
Consider injective maps
$\phi: [r]\to V(G)$.
For each edge $(u,v)$ arriving in the graph stream, we generate a virtual
stream as follows. For every edge $\{i,j\}\in E(J)$ and every injective map
$\phi:[r]\to V(G)$ satisfying
\[
\{\phi(i),\phi(j)\}=\{u,v\},
\]
we output the item $\phi$. Thus the frequency of an item $\phi$ in the virtual stream is
\[
f_J(\phi)
=
\left|\left\{\{i,j\}\in E(J):
\{\phi(i),\phi(j)\}\in E(G)\right\}\right|.
\]
Therefore, $\phi$ is a non-induced copy of $J$ exactly when
\[
f_J(\phi)=q.
\]

Now, we consider the polynomial
\[
\binom{x}{q}
=
\frac{x(x-1)\cdots(x-q+1)}{q!}.
\]
Since $\binom{x}{q}$ is a degree-$q$ polynomial, it may be written as
\[
\binom{x}{q}
=
\sum_{k=0}^{q} a_k x^k.
\]
Since $0\le f_J(\phi)\le q$, we have
\[
\binom{f_J(\phi)}{q}
=
\begin{cases}
1, & f_J(\phi)=q,\\
0, & f_J(\phi)<q.
\end{cases}
\]
Hence, the number of labelled non-induced copies of $J$ is
\[
N(J)
=
\sum_{\phi}\binom{f_J(\phi)}{q}.
\]

Therefore, we have
\[
N(J)
=
\sum_{\phi}\sum_{k=0}^{q} a_k f_J(\phi)^k
=
\sum_{k=0}^{q} a_k \FM k,
\]
where $\FM k$ is the $k$th frequency moment of the virtual stream of injective maps $\phi$ as defined
above. Thus the labelled non-induced copy count of $J$ can be recovered using appropriate coefficients $\sequence{a}{q}$, and the corresponding frequency moments
\[
\FM0,\FM1,\ldots,\FM q.
\]
Since $J$ is fixed, $q$ is constant. By \Cref{Thm: Fk in 4}, each $\FM k$ can be
computed exactly in four passes using $O(k\log n)$ clean space, and $O(k^2 n\log n)$ catalytic space. Computing all
moments up to $q$ in parallel uses
\[
\sum_{k=0}^{q} O(k\log n)
=
O(q^2\log n)
\]
clean space, $O(q^3 n\log n)$ catalytic space and four passes,. The above counts the number of non-induced copies. 

\paragraph{Counting $I(J)$:}
We now show how to count induced occurrences. Fix any graph $J$ that contains $H$ as a subgraph (i.e., $J \supseteq H$). Consider a vertex set that contributes to $N(J)$. By definition, it contains all edges of $J$. The induced subgraph on that vertex set is therefore some graph $J'$ that has $J$ as a subgraph: $J \subseteq J'$. In other words, every non-induced copy of $J$ induces a (possibly larger) supergraph $J'$ of $J$. Conversely, every induced copy of any supergraph $J' \supseteq J$ is automatically a non-induced copy of $J$. Hence we obtain:
\[
N(J) = \sum_{J' \supseteq J} I(J').
\tag{1}
\]
This identity holds for every fixed graph $J$ on $r$ vertices. This is a finite upper-triangular linear system over edge count wise sorted graphs on $r$
vertices. Since $r$ is fixed, by the general
inclusion-exclusion principle, we obtain
\[
I(H) = \sum_{J \supseteq H} (-1)^{|E(J)| - |E(H)|} \, N(J).
\tag{2}
\]

Therefore the labeled induced count of $H$ can be computed from the
non-induced counts of all supergraphs $J\supseteq H$. Observe that there are at most
\[
2^M,
\qquad
M=\binom{r}{2},
\]
such graphs, and each requires moments only up to order at most $M$.
Thus the usage of total clean space is bounded by
\[
O(2^M M^2q^2\log n),
\]
and the number of passes remains four, since all required moment
computations can be run in parallel. The catalytic space usage becomes $O(2^M M^2q^3\log n)$. Finally, to obtain unlabelled copies of $H$, it suffices to divide
labelled induced count by the size of the automorphism group of $H$.
\end{proof}

\begin{corollary}[Exact triangle counting]
There is an exact $3$-pass catalytic-space algorithm for counting
triangles in an $n$-vertex graph stream using $O(\log n)$ space.
\end{corollary}

\begin{proof}
For triangle counting, we use the virtual stream in which every arriving
edge $(u,v)$ generates the triples
\[
\{u,v,w\},
\qquad w\in V\setminus\{u,v\}.
\]
If a triple has frequency $i$, then it induces exactly $i$ edges. Let
$T_i$ be the number of triples with frequency $i$. Then
\[
\FM k=T_1\cdot 1^k+T_2\cdot 2^k+T_3\cdot 3^k.
\]
Thus
\[
\#\triangle(G)=T_3
=
\frac{\FM 3-3\FM 2+2\FM 1}{6}.
\]
By assumption, $\FM 2$ can be computed exactly in two passes using
$O(\log n)$ space, and $\FM 3$ can be computed exactly in three passes using
$O(\log n)$ space. Also $\FM 1$ is trivial to compute in one pass. Running
these computations in parallel gives an exact three-pass algorithm using
$O(\log n)$ space.
\end{proof}


\section{Lower Bounds for 1-pass Catalytic Streaming Algorithms} \label{sec:lower}

In this section, we show that catalytic space provides no advantage in the one-pass streaming model. In particular, any problem that can be solved by a deterministic (randomized) one-pass streaming algorithm using \(s\) bits of workspace and \(c\) bits of catalytic space can also be solved by a standard deterministic (respectively, randomized) one-pass streaming algorithm using only \(s\) bits of workspace. To prove this, we use an automata-theoretic view of streaming algorithms. The proof is inspired by the proof that in the traditional setting Catalytic LOGSPACE is contained in the probabilistic class ZPP~\cite{Buhrman/STOC/2014/CatalyticIntro}.




\begin{remark}
   We note that our lower bound results for one-pass catalytic streaming algorithms continue to hold even when the stream length, $\streamlen$, is known in advance. 
\end{remark}

\begin{remark}
Space usage in streaming algorithms can be categorized into two distinct types. The first, which we refer to as {\emph storage space}, is the amount of memory used to maintain information between the arrivals of consecutive items in the stream. This is typically what is meant by the {\emph space complexity} of a streaming algorithm. The second type is the additional memory required during the processing of an incoming item, which we call {\emph{update space}}. This corresponds to the temporary workspace needed to update the stored summary upon each arrival.
Ideally, one would define the overall space complexity of a streaming algorithm as the sum of its storage and update space. However, most existing lower bound techniques in the streaming literature—such as those based on communication complexity \cite{Roughgarden/NOWPub/2016/CommunicationComplexity} or automata-theoretic arguments \cite{PavanCVM/PODS/2024/ForgetDataStream} provide lower bounds on the storage space, and irrespective of the amount of the update space used.
Since any lower bound on storage space immediately implies a lower bound on the total space complexity, this distinction is often overlooked. Nevertheless, for the purposes of our results, it is important to be precise.

Throughout this paper, unless stated otherwise, the term space complexity refers exclusively to the storage space complexity, and does not include update space.




\end{remark}



\subsection{Deterministic-Streaming-Automaton}

We define the following {automata}, which we call the \emph{streaming-automata}. This will help us model the deterministic (and randomized) catalytic (and non-catalytic) one-pass streaming algorithms.  We note that the deterministic non-catalytic version of the {\emph streaming-automata} is basically a Moore machine. We start with defining the same. 

\begin{definition}
    A \emph{simple deterministic streaming-automata} is a 6-tuple 
    \[\fbrac{Q,q_0, \streamuniverse, \func{\delta}{Q\times\streamuniverse}{Q},\func{O}{Q}{\sO}}\]
    consisting of a state space $Q$,  a state $q_0$,  an universe (or alphabet) $\streamuniverse$, a transition function $\func{\delta}{Q\times\streamuniverse}{Q}$ and an output function $\func{O}{Q}{\sO}$, where $\sO$ is an output space. 

    \

    We denote the composition of transition rules $\func{\Delta}{Q\times\streamuniverse}{\streamuniverse}$ defined as $$\Delta(q,x\omega) = \delta(\Delta(q,x),\omega).$$    For a stream $\stream \in \streamuniverse^{\streamlen}$ the output of the automaton is defined as $O(\Delta(q_0, \stream))$.
\end{definition}

The catalytic version of the streaming-automaton has $Q\times C$ as the state-space, where the $C$ represents the catalytic state space. On an input $\stream$ the start state is $(q_0, c)$, where $c\in C$ is a catalytic state choose by an adversary. The automata has to ensure that the output function $\func{O}{(Q\times C)}{\sO\times C}$ has to recover the catalytic state. The formal definition is as follows:

\begin{definition}
    A \emph{catalytic-deterministic-streaming-automata} is a 6-tuple 
    \[\fbrac{(Q\times C),q_0, \streamuniverse, \func{\delta}{(Q\times C)\times\streamuniverse}{(Q\times C)},\func{O}{(Q\times C)}{\sO\times C}}\]
    consisting of a state space $Q\times C$ (with $Q$ representing the clean original state space and $C$ representing the catalytic state space),  a state $q_0\in Q$,  an universe (or alphabet) $\streamuniverse$, a transition function $\func{\delta}{(Q\times C)\times\streamuniverse}{(Q\times C)}$ and an output function $\func{O}{(Q\times C)}{\sO\times C}$, such that for any stream $\stream$ and any catalytic state $c\in C$,
    $$O(\Delta((q_0, c), \stream)) \in (\sO\times \{c\}).$$

    Here the $\Delta$ is defined as above. 

\end{definition}


\subsection{Lower Bounds for 1-pass Deterministic Catalytic Streaming Algorithms}

We first set up an equivalence between one-pass streaming algorithms and streaming automata. 

\begin{claim}\label{Clm: Det Stream Model Equivalence}
    There exists an one-pass deterministic streaming algorithm $\sA$ with storage space $s$ outputting values in $\sO$, and working on a stream where item comes from $\streamuniverse$  if and only if there is a simple deterministic streaming-automaton $\fbrac{Q,q_0,\streamuniverse,\func{\delta}{Q\times\streamuniverse}{Q},\func{O}{Q}{\sO}}$
    such that $\size{Q} = 2^s$ and  for all $\stream \in \Omega^*$ the output of the automaton on input $\stream$ is same as the output of the algorithm in the same input, that is,   $\sA(\stream) = O(q_0,\Delta(\stream))$.
\end{claim}

\begin{proof}

We show that any algorithm that can be implemented by a streaming automata can be implemented as a one-pass streaming algorithm, and vice-versa. The proof is standard, but we briefly describe it here. 


\vspace{2mm}
\noindent{\em Streaming Algorithm to Automata:}
    A streaming algorithm using $s$ bits can have at most $2^s$ possible configurations in its memory. Let $Q$ be the set of all possible $2^s$ configurations.
    Before starting to process the stream, the algorithm performs an initialization step. The starting state $q_0$ is the the state corresponding to the memory configuration after the initialization step.
    At each arrival of an element $\omega \in \streamuniverse$ in the stream, the streaming algorithm performs an update operation. Any update operation can be seen as the changes made in the memory available to the algorithm. Hence, each update step can be modelled using a transition function $\func{\delta}{Q\times\streamuniverse}{Q}$.
   At the end of  the stream the algorithm outputs a value from $\sO$ based on the configuration of the memory. Thus the output function is the $O$ is the mapping from the configuration set to the $\sO$ that the algorithm does.

   Thus for any stream $\stream$ the output of the algorithm $\sA$ is the same as that of the automata.


\vspace{2mm}

\noindent{\em Automata to Streaming Algorithm:}
    Given a automata $\fbrac{Q,q_0, \streamuniverse, \func{\delta}{Q\times\streamuniverse}{Q},\func{O}{Q}{\sO}}$ we name the states in $Q$ using $(\log |Q|)$ bits. In the streaming algorithm we will use the storage space to store the name of the current state.  
    The initialization step sets the storage space stores the name of the state $q_0$ in the memory. At the arrival of each element, the update step updates the memory space as defined by the transition function of the automata, that is, on the arrival of item $\omega$ the algorithm replaces the name of the current state $q$ to the name of $\delta(q, \omega)$. At the end, the algorithm outputs according to the output function $O$ and the current state $q \in Q$. 
    Clearly the output of the algorithm is same as the output of the automata on input $\stream$.
\end{proof}

A similar proof proves the following fact.

\begin{claim}\label{Clm: Det Cat Space Model Equivalence}
 There exists an one-pass deterministic streaming algorithm $\sA$ with storage space $s$ and catalytic space $c$ outputting values in $\sO$, and working on a stream where item comes from $\streamuniverse$  if and only if there is a catalytic-deterministic-streaming-automaton 
    $$\fbrac{(Q\times C),q_0,\streamuniverse,\func{\delta}{(Q\times C)\times\streamuniverse}{(Q\times C)},\func{O}{(Q\times C)}{\sO\times C}},$$
    such that $\size{Q} = 2^s$, $\size{C} = 2^c$ and  for all $\stream \in \Omega^*$ the output of the automaton on input $\stream$ is same as the output of the algorithm on the same input, that is,   $\sA(\stream) = O(q_0,\Delta(\stream))$.

\end{claim}


The proof is very similar to the proof of Claim~\ref{Clm: Det Stream Model Equivalence}. We give the proof of only one direction (only if direction) of the Claim~\ref{Clm: Det Cat Space Model Equivalence}, since that is the direction we need for our later proofs. 

\begin{proof}

We show that any algorithm that can be implemented by the automata can be implemented as a one-pass streaming algorithm.

\vspace{2mm}
\noindent{\em Streaming Algorithm to Automata:}
    A streaming algorithm using $s$ bits can have at most $2^s$ possible configurations in its memory. Similarly, $c$ bits of catalytic space can have $2^c$ possible configurations. We map each of the storage and catalytic memory configurations to one of the states in $Q$ and $C$, respectively. 
    
    Before starting to process the stream, the algorithm performs an initialization step. Let us denote the state corresponding to the storage memory configuration after the initialization step as $q_0$. Hence, the set of possible start states can be defined as $S_0 = \sbrac{(q_0,c)|\forall c \in C}$.

    At the arrival of an element $\omega \in \streamuniverse$ in the stream, a streaming algorithm performs an update operation. Any update operation can be seen as the changes made in the storage space and the catalytic space available to the algorithm. Hence, each update step can be modelled using a transition function $\func{\delta}{Q\times C\times\streamuniverse}{Q\times C}$.

    This ensures that the model at the end of processing a stream $\stream$ is at the state corresponding to the final storage memory and catalytic memory configuration of the streaming algorithm $\sA$. Then, the output function can then output same as $\sA(x)$ as the algorithm only outputs based on its stored memory.

\end{proof}

\begin{claim}\label{Clm: Det Cat No Intersection}
    Let $Q_y = \sbrac{(q,z)|\Delta((q_o,y),x), x \in \streamuniverse^*}$ denote the set of possible states reachable from $(q_o,y)$ on any input. Then, for $y_1 \neq y_2$, $Q_{y_1} \cap Q_{y_2} = \emptyset$.
\end{claim}

\begin{proof}
    For contradiction, let us assume that there exists streams $\stream_1$ and $\stream_2$ (not necessarily different) such that for $y_1 \neq y_2$, $\Delta((q_0,y_1),\stream_1) = \Delta((q_0,y_2),\stream_2)$. Let us denote this common state as $(q,y_3)$.
    Consider the output function $O$ on the state $(q,y_3)$. Clearly, $O((q, y_3))$ cannot be both $y_1$ and $y_2$, implying that either the output of the automata on stream $\stream_1$ violates the necessary condition of the output function or  the output of the automata on stream $\stream_2$ violates the necessary condition of the output function. Hence, a contradiction.
\end{proof}

\begin{claim}\label{Clm: Det Small Cat Model Exists}
    $\exists y \in C$ such that $\size{Q_y} \leq 2^s$.
\end{claim}

\begin{proof}
    Note that $\size{Q\times C} = 2^{s+c}$, and $\size{C} = 2^c$. From Claim~\ref{Clm: Det Cat No Intersection} we have that all the sets $Q_y$ are disjoint. So by pigeon-hole-principle there exists a $y$ with $\size{Q_y} \leq 2^s$.
\end{proof}

\begin{theorem}\label{thm:lbdet}
    If there exists a deterministic one-pass catalytic streaming algorithm using space $s$, then there exists a deterministic one-pass streaming algorithm using space $s$ that gives the same output as the one-pass catalytic streaming algorithm on any input stream.
\end{theorem}

\begin{proof}
    Suppose there exists such a catalytic streaming algorithm $\sA$ that uses $s$ bits of clean space and $c$ bits of catalytic space. By Claim~\ref{Clm: Det Cat Space Model Equivalence} there exists a catalytic-deterministic-streaming-automata performing exactly as the algorithm $\sA$.  Let this automata be $$\fbrac{(Q\times C),q_0,\streamuniverse,\func{\delta}{(Q\times C)\times\streamuniverse}{(Q\times C)},\func{O}{(Q\times C)}{\sO\times C}},$$ with $|Q|=2^s$ and $|C|=2^c$.    
    In this automata, by Claim~\ref{Clm: Det Small Cat Model Exists}, there exists a $y \in C$ such that $\size{Q_y} \leq 2^s$. Also, from Claim~\ref{Clm: Det Cat No Intersection} we know that  these states are not reachable from any $(q_0, z)$ for any $z\neq y$. 
    Also, note that for all states in $Q_y$ the output function $O$ outputs a value in $\sO\times \{y\}$.
    
    Now, consider the states $Q_y$ and using the transition function $\delta$ and output function $O$ on the states $Q_y$ we obtain a simple-deterministic-streaming automata
    $$\fbrac{Q_y,(q_0,y),\streamuniverse,\func{\delta}{Q_y\times\streamuniverse}{Q_y},\func{O}{Q_y}{\sO}},$$ with $|Q_y|=2^s$, where the new output function only returns the first term of the what the previous output function (by ignoring the catalytic space part in the output).  The output of this automata on any stream $\stream$ is exactly the same as the output of the algorithm $\sA$.

    Now by Claim~\ref{Clm: Det Stream Model Equivalence} there is an equivalent algorithm that uses $\log |Q_y|$ bits of space giving the same output as the catalytic algorithm $\sA$.

\end{proof}

\subsection{Randomized streaming automaton}

We can generalize for the catalytic and non-catalytic versions of the streaming-automaton to the randomized setting, where the transitions are randomized.  Let us start with defining the simple-randomized streaming automaton.

\begin{definition}
    A \emph{simple randomized streaming-automata} is a 6-tuple 
    \[\fbrac{Q,q_0, \streamuniverse, \func{\randtran}{\streamuniverse}{[0,1]^{|Q|\times |Q|}},\func{O}{Q}{\sO}}\]
    consisting of a state space $Q$,  a state state $q_0$,  an universe (or alphabet) $\streamuniverse$, a transition matrix for all element of the universe $\func{\randtran}{\streamuniverse}{[0,1]^{|Q|\times |Q|}}$ and an output function $\func{O}{Q}{\sO}$, where $\sO$ is an output space.  Moreover, the transition matrix has the following property:
    
    For any $\omega \in \Omega$, we denote by $\randtran^\omega$ the transition matrix $\randtran(\omega)$, and for any $q', q\in Q$ we denote by
    $\randtran^\omega[q', q]$ the $q',q$-th entry in the transition matrix. This entry represents the probability of going from state $q'$ to state $q$ on input $\omega$. Thus for each input $\omega \in \Omega$ $\randtran^\omega$ satisfies  $$\sum_{q \in Q} \randtran^\omega[q',q] = 1.$$
    We can naturally view $\randtran$ as a function from $[0,1]^{\size{Q}}\times \streamuniverse$ to $[0,1]^{\size{Q}}$, as 
    $\randtran(v) = \randtran^{\omega}.v^{T}$     
     We also denote by $\func{\Randtran}{[0,1]^{\size{Q}}\times\streamuniverse^*}{[0,1]^{\size{Q}}}$ the composition of the transition rule defined as $\Randtran(e_{q},x\omega) = \randtran^\omega(\Randtran(e_{q},x))$ where $e_{q}$ is a $\size{Q}$ length unit vector with $1$ at the entry corresponding to $q$, and $0$ elsewhere.

     Thus, on any stream $\stream$ the output of the simple-randomized-streaming-automaton is a distribution over the output space $\sO$ with 
     $$\Pr[\mbox{ Output on stream $\stream$ is $o$}] = \sum_{q|O(q) = o}\Randtran(e_{q_0},\stream)[q].$$

\end{definition}

Now, as in the case of deterministic streaming-automata we extend the definition of simple-randomized-streaming-automata to catalytic-randomized-streaming automaton. For completeness we give the formal definition below:

\begin{definition}
    A \emph{catalytic-randomized-streaming-automata} is a 6-tuple 
    \[\fbrac{(Q\times C),q_0, \streamuniverse, \func{\randtran}{\streamuniverse}{[0,1]^{|Q\times C|\times |Q\times C|}},\func{O}{Q\times C}{\sO}}\]
    consisting of a state space $Q\times C$,  a state state $q_0\in Q$,  an universe (or alphabet) $\streamuniverse$, a transition matrix for all element of the universe $\func{\randtran}{\streamuniverse}{[0,1]^{|Q\times C|\times |Q\times C|}}$ and an output function $\func{O}{Q\times C}{\sO}$, where $\sO$ is an output space.  Moreover, the transition matrix has the following property:
    
    For each input $\omega \in \Omega$ and every $q'\in Q\times C$ the transition function $\randtran^\omega$ satisfies  $$\sum_{q \in Q\times C} \randtran^\omega[q',q] = 1.$$ 
    Moreover, for any stream $\stream$ and any catalytic state $c\in C$, and any $q\in Q\times C$
    $$ \mbox{If } \Randtran((q_0, c)[q] \neq 0 \mbox{ then }
    O(q) \in (\sO\times \{c\}),$$
where $\Randtran$ is the natural extension of $\randtran$ to streams.  

Note that, the last condition ensures that irrespective of how low the probability of reaching a state may be the output always recovers the catalytic state.

\end{definition}

\subsection{Lower Bounds for 1-pass Randomized Catalytic Streaming Algorithms}

The following claim proves the equivalence between 1 pass randomized streaming algorithm and simple-randomized-catalytic-streaming-automaton. The proof is identical to the the proof in the deterministic case. So we skip it.


\begin{claim}\label{Clm: Rand Stream Model Equivalence}
 There exists an one-pass randomized streaming algorithm $\sA$ with storage space $s$ outputting values in $\sO$, and working on a stream where item comes from $\streamuniverse$  if and only if there is a simple-deterministic-streaming-automaton 
 $\fbrac{Q,q_0, \streamuniverse, \func{\randtran}{\streamuniverse}{[0,1]^{|Q|\times |Q|}},\func{O}{Q}{\sO}}$
 such that $\size{Q} = 2^s$ and  for all $\stream \in \Omega^*$ the output distribution of the automaton on input $\stream$ is the same as the output distribution of the algorithm in the same input.
\end{claim}





Similarly, we can prove that 1 pass randomized streaming algorithm using catalytic space is equivalent to randomized-catalytic-streaming-automaton.

\begin{claim}\label{Clm: Rand Cat Space Model Equivalence}
 There exists an one-pass randomized streaming algorithm $\sA$ with storage space $s$ and $c$ bit of catalytic space  outputting values in $\sO$, and working on a stream where item comes from $\streamuniverse$  if and only if there is a catalytic-randomized-streaming-automaton 
 $$\fbrac{(Q\times C),q_0, \streamuniverse, \func{\randtran}{\streamuniverse}{[0,1]^{|Q\times C|\times |Q\times C|}},\func{O}{Q\times C}{\sO}\times C}$$
 such that $\size{Q} = 2^s$, $\size{C} = 2^c$ and  for all $\stream \in \Omega^*$ the output distribution of the automaton on input $\stream$ is the same as the output distribution of the algorithm in the same input.
\end{claim}

\begin{claim}\label{Clm: Rand Cat No Intersection}
    Let $Q_c = \sbrac{(q',c')|\Randtran(e_{(q_0,c)},x)[q',c'] > 0,x\in\streamuniverse^*}$ denote the set of states reachable from $(q_0,c)$ on any input. Then, for $c_1 \neq c_2$, $Q_{c_1} \cap Q_{c_2} = \emptyset$.
\end{claim}


\begin{proof}
    For contradiction, let us assume that there exists $\stream_1$ and $\stream_2$ (not necessarily different) such that for some $c_1 \neq c_2$ and $(q,c_3)$, we have $\Randtran(e_{(q_0,c_1)},x_1)[q,c_3] > 0$ and $\Randtran(e_{(q_0,c_2)},x_2)[q,c_3] > 0$. Consider the output function $O$ on the state $(q,c_3)$. Clearly, $O((q,c_3))$ can not be both $c_1$ and $c_2$. Hence, the model, starting with catalytic spaces $c_1$ and $c_2$, and running on strings $x_1$ and $x_2$ has non-zero probability of outputting $c_3$ in the catalytic space. This violates the necessary condition of the output function. Hence, a contradiction.
\end{proof}

\begin{claim}\label{Clm: Rand Small Cat Model Exists}
    $\exists y \in C$ such that $\size{Q_y} \leq 2^s$.
\end{claim}

\begin{proof}
    Note that $\size{Q\times C} = 2^{s+c}$, and $\size{C} = 2^c$. From Claim~\ref{Clm: Det Cat No Intersection} we have that all the sets $Q_y$ are disjoint. So by pigeon-hole-principle there exists a $y$ with $\size{Q_y} \leq 2^s$.
\end{proof}

\begin{theorem}\label{thm:lbrand}
    If there exists a randomized one-pass catalytic streaming algorithm using space $s$, then there exists a randomized one-pass streaming algorithm using space $s$ whose output distribution on any input stream is the same as the output distribution of the one-pass catalytic streaming algorithm on the same input stream.
\end{theorem}
    
\begin{proof}
    Suppose there exists such a randomized catalytic streaming algorithm $\sA$ that uses $s$ bits of clean space and $c$ bits of catalytic space. By Claim~\ref{Clm: Rand Small Cat Model Exists}, there exists a catalytic randomized-automaton performing exactly as the algorithm $\sA$. Let this automata be
    \[\fbrac{(Q\times C),q_0, \streamuniverse, \func{\randtran}{\streamuniverse}{[0,1]^{|Q\times C|\times |Q\times C|}},\func{O}{Q\times C}{\sO}}\]
    with $|Q|=2^s$ and $|C|=2^c$. 

    n this automata, by Claim~\ref{Clm: Det Small Cat Model Exists}, there exists a $c \in C$ such that $\size{Q_y} \leq 2^s$. Also, from Claim~\ref{Clm: Det Cat No Intersection} we know that  these states are not reachable from any $(q_0, z)$ for any $z\neq y$. 
    Also, note that for all states in $Q_y$ the output function $O$ outputs a value in $\sO\times \{y\}$.

    Now, consider the states $Q_y$ and using the transition function $\randtran$ and output function $O$ on the states $Q_y$ we obtain a simple-randomized-streaming automata
    \[\fbrac{Q_y,(q_0,c), \streamuniverse, \func{\randtran}{\streamuniverse}{[0,1]^{|Q_y|\times |Q_y|}},\func{O}{Q_y}{\sO}}\] with $|Q_y|=2^s$, where the new output function only returns the first term of the what the previous output function (by ignoring the catalytic space part in the output).  The output of this automata on any stream $\stream$ is exactly the same as the output of the algorithm $\sA$.
    
    Now by Claim~\ref{Clm: Det Stream Model Equivalence} there is an equivalent algorithm that uses $\log |Q_y|$ bits of space giving the same output as the catalytic algorithm $\sA$.
\end{proof}

The lower bound on space complexity for computing $\FM{k}$ ($k\neq 1$) in one-streaming (deterministic or randomized) algorithm is $\Omega(n)$. The lower bound follows from the disjointness problem in communication complexity~\cite{RAZBOROV1992385}. Thus, from Theorem~\ref{thm:lbrand} we have Theorem~\ref{thm:1passFk}.

\section*{Acknowledgments}

We thank an anonymous researcher for pointing out that a four-pass algorithm for computing $\FM{k}$ can be obtained by adapting the catalytic register-program construction of ~\cite{Buhrman/STOC/2014/CatalyticIntro}. 
\newpage
\bibliographystyle{abbrvnat}
\bibliography{refs} 


\end{document}